\documentclass[reqno,11pt]{amsart}

\numberwithin{equation}{section}

    \title{Construction of Currents in Causal Fermion Systems}

    \author[F.\ Finster]{Felix Finster}
    \address{Fakultät für Mathematik \\ Universität Regensburg \\ D-93040 Regensburg \\ Germany}
    \email{finster@ur.de}

    \author[P.\ Fischer]{Patrick Fischer \\ \\ July 2025}
    \address{Fakultät für Mathematik \\ Universität Regensburg \\ D-93040 Regensburg \\ Germany}
    \email{patrick.fischer@ur.de}

\usepackage[utf8]{inputenc}

\usepackage{amsmath}
\usepackage{amsfonts}
\usepackage{amssymb}
\usepackage{amsthm}
\usepackage{amscd}

\usepackage{authblk}

\usepackage{physics}
\usepackage{tensor}
\usepackage{slashed}

\usepackage{bbold}
\usepackage[mathscr]{eucal}
\usepackage{mathtools}

\usepackage[nolist]{acronym}
\usepackage[linktocpage]{hyperref}
\usepackage[font=small,labelfont=bf]{caption}

\usepackage{graphicx}
\usepackage{pgfplots}
\usepackage{stmaryrd}

\usepackage{biblatex}

\pgfplotsset{compat=newest}

\DeclareMathOperator{\supp}{supp}

\DeclareMathOperator{\Span}{span}

\DeclareMathOperator{\Img}{Img}

\newcommand{\isoeq}{\cong}
\newcommand{\defeq}{\coloneqq}
\newcommand{\eqdef}{\eqqcolon}
\newcommand{\isospectralto}{\asymp}
\newcommand{\containsContribution}{\asymp}

\newcommand{\diracad}[1]{\overline{#1}}
\newcommand{\complexconj}[1]{\overline{#1}}
\newcommand{\dAlembertian}{\Box}
\newcommand{\disjointUnion}{\dot{\bigcup}}

\newcommand{\spinket}[1]{\left| #1 \succ \right.}
\newcommand{\spinbraket}[2]{{\prec #1 | #2 \succ}}

\newcommand{\N}{\mathbb{N}}
\newcommand{\R}{\mathbb{R}}
\newcommand{\C}{\mathbb{C}}
\renewcommand{\H}{\mathcal{H}}
\newcommand{\F}{\mathcal{F}}
\renewcommand{\L}{\mathcal{L}}

\numberwithin{equation}{section}

\newtheorem{definition}{Definition}[section]
\newtheorem{theorem}[definition]{Theorem}
\newtheorem{proposition}[definition]{Proposition}
\newtheorem{lemma}[definition]{Lemma}
\newtheorem{corollary}[definition]{Corollary}

\begin{acronym}
    \acro{el}[EL]{Euler-Lagrange}
\end{acronym}

\begin{document}

    \maketitle

    \begin{abstract}
        This paper presents a novel and systematic formalism for deriving classical field equations within the framework of
        causal fermion systems,
        explicitly accounting for higher-order corrections such as quantum effects and those arising from spacetime discreteness.
        Our method, which also generalizes to non-abelian gauge fields and gravitation, gives a systematic procedure for evaluating the
        linearized field equations of causal fermion systems.
        By probing these equations with specific wave functions and employing Taylor expansions, we reformulate them as a
        family of tensorial equations of increasing rank.
        We show that, for rank one, this approach recovers the established classical dynamics corresponding to Maxwell's equations.
        In addition, the approach gives rise to higher-rank tensorial equations, where the second-rank equations are expected to
        encode the Einstein equations, and higher-rank tensors potentially reveal new physics and systematic corrections.
    \end{abstract}

    \tableofcontents

    \section{Introduction}\label{sec:01-introduction}
    The quest for a unified theory of fundamental physics begins with formulating the laws of nature through a variational principle.
    The theory of causal fermion systems offers a novel and mathematically rigorous approach, grounding the formulation of physical
    equations in a fundamental action principle: the \textit{Causal Action Principle}~\cite{intro}.
    This principle extends beyond classical actions, naturally leading to second-quantized equations and potentially incorporating
    Planck-scale corrections, thus offering a deeper understanding of reality that transcends classical field theories.

    Previous work, notably elaborated in~\cite{cfs}, has explored the connection between causal fermion systems and classical field equations.
    However, these connections have not yet provided a systematic framework for deriving higher order corrections
    originating from the quantum nature of spacetime.
    Furthermore, the existing approaches lack a direct link to the abstract analytical structures presented in~\cite{nonlocal},
    which are crucial for describing more general spacetime geometries, including discrete ones.
    Overcoming these limitations is paramount for a systematic study of corrections to the classical field equations.
    These higher order corrections will not only allow for comparison with other effective field theory corrections, but also
    are a way for testing causal fermion systems in experiments.

    To achieve this, our approach leverages the linearized field equations of causal fermion systems.
    By probing these equations with specific wave functions (e.g. $\Psi(x)^*$) and employing Taylor expansions, as detailed in later
    sections, we can reformulate them as a family of equations for higher-ranked tensors.
    This allows us to recover well-known classical equations: tensors of rank one correspond to Maxwell's equations, rank two
    tensors are expected to give rise to Einstein's equations (analogously to the results in~\cite{cfs}), and crucially,
    third and higher-ranked tensors potentially reveal new physics and systematic corrections.
    More precisely, Theorem~\ref{thm:probing} states that the probing defines an equivalence between the linearized field equations
    and the probed CFS-current at any spacetime point.
    In combination with the expansion Theorem~\ref{thm:expansion}, this leads to the equivalence of tensorial equations of the form
    \begin{align}
        \displaystyle J_x (\mathbf{v}) \partial_{\mu_1} \dots \partial_{\mu_n} \Psi(x)^* = 0,
    \end{align}
    where $J_x (\mathbf{v})$ is the CFS current to the linearized field equations.

    For illustrative purposes and to demonstrate the efficacy of our method, we systematically elaborate on the equations for tensors of rank one ($n = 1$).
    This analysis is conducted in the context of $i \varepsilon$-regularized Minkowski spacetime.
    We show that this yields the Maxwell equations
    \begin{align}
        \partial_{\mu} F^{\mu \nu}(x) = \alpha ~ \diracad{\phi(x)} \gamma^{\nu} \phi(x)
    \end{align}
    and we are able to determine the scaling behavior of the coupling constant.
    This constant is shown to be independent of the regularization length $\varepsilon$.
    Most notably, the $i \varepsilon$-regularization breaks the Lorentz invariance whereas the resulting equations
    are Lorentz invariant in the continuum limit.

    This paper is organized as follows: Section~\ref{sec:02-preliminaries} lays the groundwork by introducing the fundamental
    concepts of causal fermion systems and the causal action principle, followed by a discussion of the linearized field equations.
    Section~\ref{sec:03-cfs-current-in-the-general-setting} delves into the core methodology, defining the CFS current in a general setting
    and introducing the concept of probing (Theorem~\ref{thm:probing}).
    This section also discusses the implications of assuming a manifold structure for the causal fermion system which allows
    to extract tensorial equations.
    In Section~\ref{sec:04-iepsilon-regularlized-minkowski-spacetime}, we apply our formalism to the specific case of the $i \varepsilon$-regularized Minkowski spacetime,
    explicitly computing the corresponding perturbation for Dirac particles and vector potentials.
    Section~\ref{sec:the-continuum-limit-of-minkowski-spacetime} analyzes the continuum limit of Minkowski spacetime,
    addressing UV divergences and showing how the Maxwell equations are recovered.
    Finally, Section~\ref{sec:conclusion-and-outlook} provides a conclusion and an outlook on future research directions leading
    towards the completion of the derivation of the Standard Model and General Relativity.
    In addition, we outline how this method can be applied to derive Planck-scale corrections implied by causal fermion systems.

    \section{Preliminaries}\label{sec:02-preliminaries}

    The theory of causal fermion systems presents itself as a compelling candidate for a unifying theory of fundamental physics,
    primarily due to its inherently relational and background-independent formulation.
    At its core, causal fermion systems operates on an abstract separable complex Hilbert space $\H$.
    In contrast to quantum field theories, the vectors of $\H$ are \textit{a priori} functions or sections on a predefined manifold.
    Hence, unlike conventional physical theories that often presuppose a fixed spacetime background, causal fermion systems posits that spacetime itself,
    identified with the support $M \defeq \supp \rho$ of a measure $\rho$ on a specific set of operators $\F(\H)$ is itself a dynamical quantity.
    The spacetime emerges from the intricate `web of correlations' encoded by these operators acting on $\H$.
    This perspective, as detailed in~\cite{cfs-correlations}, asserts that the fundamental properties of the universe are
    encoded solely by these intrinsic correlations, thereby offering a complete and viable ontology that is not reliant on
    continuum limits for its physical interpretation.
    Concrete realizations of $\H$, such as function spaces on Minkowski spacetime, are only introduced in specific limiting cases,
    like the continuum limit.
    This profound independence from a pre-defined background structure and its reliance on fundamental correlational dynamics
    make causal fermion systems particularly well-suited to address long-standing challenges in quantum gravity,
    providing a consistent framework capable of describing both discrete and continuous spacetime structures.

    In this section, we introduce causal fermion systems and the causal action principle.
    We outline the derivations for causal fermion systems relevant for this paper .
    For a more detailed analysis, the reader is referred to the textbooks~\cite{intro, cfs}

    \subsection{Causal Fermion Systems}\label{subsec:causal-fermion-systems}

    As stated above, the theory of causal fermion systems starts with an abstract Hilbert space $\H$.
    On $\H$ we consider operators $\F \subset L(\H)$, then the physical object of a causal fermion system is a measure on $\F$.
    More concretely, we have the following mathematical definition.

    \begin{definition}[Causal Fermion System]
        Let $(\H, \braket{\cdot}{\cdot})$ be a separable complex Hilbert space, $n \in \mathbb{N}$
        and $\F_n \subset L(\H)$ be the set of all symmetric operators on $\mathcal{H}$ of finite rank, which
        (counting multiplicities) have at most $n$ positive and at most $n$ negative eigenvalues.
        Further, let $\rho$ be a regular Borel measure (defined on a $\sigma$-algebra of $\F_n$).
        We refer to $(\H, \F_n, \rho)$ as a \textbf{causal fermion system} of \textbf{spin dimension} $n$.
    \end{definition}

    Already from this definition, we see that a causal fermion system encodes much information.
    First, we refer to the support\footnote{
        The support of a measure $\rho$ is defined as the complement of the largest open set of measure zero, i.e.
        \begin{align*}
            \supp \rho \defeq \F_n \setminus \bigcup \big\{ \Omega \subset \F_n : \Omega \text{ open and } \rho(\Omega) = 0 \big\}.
        \end{align*}
    } $M \defeq \supp \rho$ of the measure $\rho$ as the spacetime of a causal fermion system.
    The reason for this name becomes clear in Section~\ref{sec:04-iepsilon-regularlized-minkowski-spacetime} when we analyze a
    concrete realization of $\H$.
    Considering $\rho$ as physical quantity allows variations of spacetimes $M$.
    Since there are no requirements on the spacetime, causal fermion systems allow for variations in spacetime topologies in
    the action principle.
    In particular, the considerations in the abstract setting hold for both discrete and continuous spacetimes.

    In addition, every spacetime point $x \in M$ determines a subspace $S_x M \defeq x(\H)$, the spin space at $x$.
    Since $x$ has finite rank, $\dim S_x M \leq 2n$.
    We equip every spin space with a spin inner product defined by
    \begin{align}
        \label{eq:def-spin-inner-product}
        \spinbraket{u}{v}_x \defeq -\bra{u} x \ket{v} && \text{ for } u, v \in \H.
    \end{align}
    Having a spin space at every point in $M$ gives rise to a (not necessary smooth) fiber bundle $SM \defeq \disjointUnion_{x \in M} S_x M$.
    In particular, if we consider only operators with rank $2n$, so-called regular operators, then $SM$ is a
    vector bundle with fiber $\C^{2n}$.

    A priori, a vector $u \in \H$ does not encode information about the spacetime.
    However, we use $u$ to define a physical wavefunction
    \begin{align}
        \spinket{\psi^u(x)} \defeq \pi_x \ket{u} \in S_x M && \text{ for } x \in M,
    \end{align}
    where $\pi_x$ is the orthogonal projection onto the image of $x$.
    By construction, $\psi^u$ is a section of the fiber bundle $SM$ for every $u \in \H$.
    Additionally, we introduce the wave evaluation operator
    \begin{align}
        \Psi(x): \H \to S_x M && \Psi(x) \ket{u} \defeq \spinket{\psi^u(x)} && \text{ for } x \in M.
    \end{align}
    It is clear that $\Psi(x) = \pi_x$ and its adjoint is determined by
    \begin{align}
        \spinbraket{\varphi}{\Psi(x) u}_x &= -\bra{\varphi} x \Psi(x) \ket{u} = -\braket{x \varphi}{u} & \text{ for } \varphi \in S_x M \text{ and } u \in H \nonumber \\
        &\Rightarrow \Psi(x)^* = -x|_{S_x M} : S_x M \to \H.
    \end{align}
    Moreover, we can express the operator $x$ in terms of the wave evaluation operator as
    \begin{align}
        \label{eq:local-correlation-operator}
        x = -\Psi(x)^* \Psi(x) : \H \to \H.
    \end{align}

    Another important object, which is relevant for the computation in the rest of the paper, is the kernel of the fermionic projector
    \begin{align}
        \label{eq:def-kernel-of-the-fermionic-projector}
        P(x, y) \defeq -\Psi(x) \Psi(y)^*: S_y M \to S_x M.
    \end{align}
    It is the causal fermion system analog of the two point correlation function in quantum field theory.
    The kernel of the fermionic projector is symmetric (with respect to the spin inner products) in the sense that
    \begin{align}
        P(x, y)^* = P(y, x).
    \end{align}

    So far, all structures are present for every causal fermion system.
    In order to single out the physically admissible causal fermion system, we formulate an action principle in the next section.

    \subsection{The Causal Action Principle}\label{subsec:the-causal-action-principle}

    Although $x, y \in \F$ are symmetric operators, the product $xy$ is in general not symmetric.
    Therefore, the operator $xy$ has complex eigenvalues.
    Since, $x$ and $y$ are finite ranked, $xy$ also has rank $k \leq 2n$.
    We denote the non-vanishing eigenvalues of $xy$ by $\lambda^{xy}_i$ with $i = 1, \dots k$.
    For notational simplicity, from now on we consider $2n$ eigenvalues of $xy$, where $\lambda^{xy}_1, \dots \lambda^{xy}_k \neq 0$
    are the non-vanishing eigenvalues of $xy$.
    For $k < 2n$, we then set $\lambda^{xy}_{k + 1} \dots \lambda^{xy}_{2n} = 0$.
    Based on the eigenvalues of the operator product $xy$, we define the causal action as
    \begin{align}
        \text{Causal action:} && S[\rho] &\defeq \iint_{\F \times \F} \L(x, y) \dd \rho(x) \dd \rho(y); \\
        \text{Causal Fermion System lagrangian:} && \L(x, y) &\defeq \frac{1}{4n} \sum_{i, j = 0}^{2n} \left( \abs{\lambda^{xy}_i} - \abs{\lambda^{xy}_j} \right)^2. \label{eq:lagrangian}
    \end{align}
    The action is varied under the following constraints
    \begin{align}
        \text{Volume constraint:} && \rho(\F) &= 1;\\
        \text{Trace constraint:} && \int_\F \tr_\H x \dd \rho(x) &= 1;\\
        \text{Boundedness constraint:} && \iint_{\F \times \F} \left( \sum_{i = 1}^{2n} \abs{\lambda^{xy}_i} \right)^2 \dd \rho(x) \dd \rho(y) &< \infty.
    \end{align}
    This variational principle is mathematically well-posed if $\H$ is finite-dimensional.
    For the existence theory and the analysis of general properties of minimizing measures we refer to~\cite{discrete, continuum}.

    The constraints can be included into the variation principle by using Lagrange multipliers.
    Then, the derivation of the \ac{el} equations is based on the following proposition.

    \begin{proposition}
        \label{prop:el-equation}
        Let $\rho$ be a minimizer of the causal action principle with Lagrange multipliers $\mathfrak{s}, \mathfrak{r}, \kappa$
        for the volume, trace and boundedness constraints.
        Then the function
        \begin{align}
            \ell(x) \defeq \int_{\F} \L(x, y) - \kappa \left( \sum_{i = 1}^{2n} \abs{\lambda^{xy}_i} \right)^2 \dd \rho(y) - \mathfrak{r} \tr_\H x - \mathfrak{s}
        \end{align}
        is minimal on spacetime $M = \supp \rho$.
        In particular, the Lagrange parameter $\mathfrak{s}$ is chosen such that $\ell|_M = 0$.
    \end{proposition}

    The proof of this proposition can be found in~\cite{nonlocal}.
    A necessary condition for this minimality of $\ell$ is that the first variation of $\ell(x)$ vanishes on $M$, i.e.\
    \begin{align}
        \label{eq:restricted-el-equations}
        \ell(x) = 0 &&\text{ and }&& D \ell(x) = 0 && \text{ for } x \in M \:,
    \end{align}
    referred to as the {\em{restricted \ac{el} equations}}.
    Note, that even tough $M$ might have a discrete structure, $D \ell$ is well-defined as a derivative in $\F$.
    To compute $D \ell$, we make use of~\eqref{eq:local-correlation-operator}.
    Thus, a perturbation $\delta x$ of a spacetime point can be expressed as
    \begin{align}
        \label{eq:spacetime-point-variation}
        \delta x = -D_\mathbf{u} \Psi(x)^* \Psi(x) - \Psi(x)^* D_\mathbf{u} \Psi(x),
    \end{align}
    where $\mathbf{u}$ is the corresponding tangent vector at $x$ in $\F$.
    In addtion, the operator products $AB$ and $BA$ have the same non-zero eigenvalues.
    Therefore,
    \begin{align}
        xy = \Psi(x)^* \Psi(x) \Psi(y)^* \Psi(y) &\isospectralto \Psi(x) \Psi(y)^* \Psi(y) \Psi(x)^* \nonumber \\
        &= P(x, y) P(y, x) \eqdef A_{xy},
    \end{align}
    where $\isospectralto$ means isospectral (in the sense that they have the same non-trivial eigenvalues with the same
    algebraic multiplicities).
    The operator $A_{xy}$, referred to as \textit{closed chain}, is a symmetric operator on the spin space $S_x M$.
    As stated above, the operator $A_{xy}$ and $xy$ have the same, non-vanishing eigenvalues.
    Thus, the Lagrangian $\L$ can be treated as a function of $A_{xy}$ and therefore also of $P(x, y)$.
    The first variation of $\L$ is then given by
    \begin{align}
        \label{eq:first-variation-of-the-lagrangian}
        \delta \L(x, y) = 2 \Re \tr_{S_x M} \left[ Q(x, y) \delta P(x, y)^* \right].
    \end{align}
    Every perturbation $\delta P(x, y)$ is again symmetric in a sense that $\delta P(x, y)^* = \delta P(y, x)$.
    This can be directly seen from the fact that $P(x, y)$ is symmetric or one applies the product rule
    to~\eqref{eq:def-kernel-of-the-fermionic-projector}.
    In addition, since $\L$ is symmetric under the change of its arguments, i.e. $\L(x, y) = \L(y, x)$,
    the kernel $Q(x, y)$ must also be symmetric, i.e.~$Q(x, y)^* = Q(y, x)$.

    Using~\eqref{eq:def-kernel-of-the-fermionic-projector}, we get
    \begin{align}
        D_{1, \mathbf{u}} \L (x, y) = 2 \Re \tr_{S_x M} \left[ Q(x, y) \Psi(y) D_\mathbf{u} \Psi(x)^* \right].
    \end{align}
    where $D_{1, \mathbf{u}}$ denotes a derivative of the first argument of $\L$ in direction $\mathbf{u}$.
    Similarly, the first variation of the trace of $x$ is given by
    \begin{align}
        D_{\mathbf{u}} \tr_{\H} x = 2 \Re \tr_{S_x M} \left[ \Psi(x) D_\mathbf{u} \Psi(x)^* \right],
    \end{align}
    Therefore, the minimality condition of $\ell$ on $M$ implies that
    \begin{align}
        2 \Re \int_\F \tr_{S_x M} \left[ Q(x, y) \Psi(y) D_\mathbf{u} \Psi(x)^* \right] \dd \rho(y) = 2 \mathfrak{r} \Re \tr_{S_x M} \left[ \Psi(x) D_\mathbf{u} \Psi(x)^* \right]
    \end{align}
    for all $\mathbf{u}$.
    Thus,
    \begin{align}
        \label{eq:el-equation}
        (Q\Psi)(x) \defeq \int_\F Q(x, y) \Psi(y) \dd \rho(y) = \mathfrak{r} \Psi(x).
    \end{align}
    We refer to~\eqref{eq:el-equation} as the restricted \ac{el} equations of causal fermion system.

    \begin{proposition}
        \label{prop:abstract-setting-q-and-symmetry-property}
        If $A_{xy}$ is diagonalizable, then
        \begin{align}
            \label{eq:explicit-q-expression}
            Q(x, y) = \frac{1}{n} \sum_{i,j = 1}^{2n} \left( \abs{\lambda^{xy}_i} - \abs{\lambda^{xy}_j} \right) \frac{\complexconj{\lambda^{xy}_i}}{\abs{\lambda^{xy}_i}} \Lambda^{xy}_i P(x, y),
        \end{align}
        where $\Lambda^{xy}_i$ are the eigenspace projectors corresponding to the eigenvalues $\lambda^{xy}_i$.
        In addition, $Q(x, y)$ satisfies
        \begin{align}
            \label{eq:p-q-commute}
            Q(x, y) P(y, x) = P(x, y) Q(y, x),
        \end{align}
        which makes $Q(x, y) P(y, x)$ a symmetric operator w.r.t.~$\spinbraket{\cdot}{\cdot}_x$.
    \end{proposition}
    \begin{proof}
        Under the assumption that $A_{xy}$ is diagonalizable, we can write $A_{xy} = \sum_i^{2n} \lambda^{xy}_i \Lambda^{xy}_i$.
        Then, $\delta \lambda^{xy}_i = \tr \left[ \Lambda^{xy}_i \delta A_{xy} \right]$ and the perturbation of the absolute value
        is given by
        \begin{align*}
            \delta \abs{\lambda_i} &= \Re \tr \left[ \frac{\complexconj{\lambda_i}}{\abs{\lambda_i}} \Lambda_i \delta A_{xy} \right].
        \end{align*}
        By definition, the closed chain is given by $A_{xy} = P(x, y) P(y, x)$ and hence
        \begin{align*}
            \delta A_{xy} &= P(x, y) \delta P(y, x) + \delta P(x, y) P(y, x).
        \end{align*}
        Applying the chain rule to the Lagrangian and plugging in the above expression for the perturbation of the eigenvalues
        we get the Expression~\eqref{eq:explicit-q-expression} for $Q(x, y)$.

        Since $A_{xy}$ and $A_{yx}$ have the same non-vanishing eigenvalues, we get
        \begin{align*}
            \Lambda^{xy}_i P(x, y) = P(x, y)  \Lambda^{yx}_i \Rightarrow Q(x, y) P(y, x) = P(x, y) Q(y, x),
        \end{align*}
        which implies that $P(x, y)$ and $Q(x, y)$ commute in the sense of~\eqref{eq:p-q-commute}.
    \end{proof}

    At this point, we want to emphasize that the restricted \ac{el} equations encode much more information than the corresponding
    equations of the Standard Model.
    A minimizing causal fermion system describes the complete ensemble of wavefunctions in $\H$.
    By describing the complete dynamics of all fermionic states in the system, it also encodes the complete dynamics
    of all bosonic fields.

    \subsection{Linearized Field Equations}\label{subsec:linearized-field-equation}

    The linearized field equations describe infinitesimal variations of the measure $\rho$ that preserve the restricted \ac{el} equations,
    akin to linearizing physical equations in other theories (e.g., linearized gravity).
    These equations are central to the analysis of causal fermion systems, offering both conceptual clarity on the causal
    nature of the dynamics and a computational pathway for explicit calculations.
    As discussed in~\cite{cfs-correlations} we assume the Hilbert space $\H$ to have a very
    large dimension, where most of the states serve as a background, i.e.~can be thought of being geometric in nature.
    On the other hand, the number of remaining states (having particle nature) is very small compared to the number of the background
    states.
    Thus, these states can always be treated as a small perturbation, allowing us to treat them in the linearized regime.

    To derive the linearized field equations, we consider a variation of the measure $\rho$ in the direction of a vector field $\mathbf{v}$.
    For that, the perturbation of $x \in M$ according to $\mathbf{v}$ is again given by~\eqref{eq:spacetime-point-variation}.
    For the measure related to the perturbed system to be a minimizer of the causal action, the restricted \ac{el} Equations~\eqref{eq:el-equation} have to be satisfied.
    This requirement is equivalent to~\cite{nonlocal}
    \begin{align}
        \label{eq:linearized-field-equations}
        \left( (D_{\mathbf{v}} Q)\Psi \right)(x) + (Q D_{\mathbf{v}} \Psi)(x) - \mathfrak{r} D_{\mathbf{v}} \Psi(x) = 0,
    \end{align}
    which we refer to as \textbf{linearized field equations}.
    Note that $Q$ as defined by~\eqref{eq:el-equation} is an integral operator with kernel $Q(x,y)$.
    A crucial characteristic of this fact is that the restricted \ac{el} Equations and the linearized field equations have a non-local nature.
    The non-local structure of these equations was analyzed using methods of functional analysis and Fourier analysis in~\cite{nonlocal}.
    We want to emphasize that the non-locality is short-ranged.
    This fact will also be elaborated further in Section~\ref{subsec:analysis-of-the-uv-divergence}.

    \section{The CFS Current in the General Setting}\label{sec:03-cfs-current-in-the-general-setting}

    In this paper, we systematically analyze contributions to the linearized field equations arising from various perturbations.
    In~\cite{nonlocal}, these contributions were interpreted as inhomogeneities, representing deviations from the \ac{el} equations.
    Instead of treating these as fixed external sources, we view them as distinct terms originating from different physical perturbations.
    The objective is to construct a total perturbation that is a linear combination of these contributions, which precisely
    solves the homogeneous linearized field equations.
    This systematic construction allows us to understand how different physical phenomena contribute to and collectively satisfy
    the fundamental dynamics of the system.
    For this analysis, we introduce the following notation.

    \begin{definition}[CFS Current]
        \label{def:cfs-current}
        Let $(\mathcal{H}, \mathcal{F}, \rho)$ be a causal fermion system and $x \in M$, then
        the \textbf{CFS current} $J_x: T_x \mathcal{F} \to (\mathcal{H} \to S_x M)$ is defined by
        \begin{align}
            J_x (\mathbf{v}) \defeq \left( (D_{\mathbf{v}} Q)\Psi \right)(x) + (Q D_{\mathbf{v}} \Psi)(x) - \mathfrak{r} D_{\mathbf{v}} \Psi(x)
        \end{align}
        for $\mathbf{v} \in T_x \mathcal{F}$.
    \end{definition}
    Note, that the CFS current $J_x(\mathbf{v})$ for a perturbation $\mathbf{v} \in T_x \mathcal{F}$ is an operator from $\mathcal{H} \to S_x M$.
    Thus, it encodes much more information than the usual currents from the Standard Model, where, for example, the electro-magnetic
    current at a point is described by only four real numbers.
    Hence, as we discuss in Section~\ref{sec:conclusion-and-outlook} the CFS current does not only
    determine the currents from the Standard Model but
    also the energy momentum tensor and in principle higher ranked tensors.

    In addition, the computation of the CFS current of different perturbations $\mathbf{v}_i$ allows for the construction
    of a linear combination $\mathbf{v}$ which then solves the linearized field equations~\eqref{eq:linearized-field-equations}
    of causal fermion system, i.e.~$J_x(\mathbf{v}) = 0$.
    This linear combination determines constraining equations for the different perturbations, which correspond to
    the classical equations of motion.

    Mathematically, construction of the linear combination and deducing the constraining equations is possible due to the
    following observation.

    \begin{lemma}
        Let $x \in M$, then the CFS current $J_x$ is $\C$-linear.
    \end{lemma}

    This lemma is obviously fulfilled as the CFS current only contains first derivatives.
    In the following, we outline a procedure how to compute the CFS current for a given perturbation.
    Note that all steps performed in this procedure are also linear in $J_x$ and hence linear in the perturbation.
    Therefore, the linearity also holds for the resulting expressions.

    \subsection{Probing the CFS-Current}\label{subsec:probing-the-cfs-current}

    Let $\mathbf{v} \in T_x \F$, then the perturbed causal fermion system is again a minimizer if $J_x(\mathbf{v}) = 0$ for all $x \in M$.
    This equation is, from a purely mathematical point of view, well-defined, but in order to deduce physical constraints
    for $\mathbf{v}$ we need a procedure to extract the information in form of tensorial equations.
    Therefore, we introduce the concept of probing.
    We say, the CFS current is probed at $y \in M$ by computing the operator product $J_x (\mathbf{v}) \Psi(y)^*: S_y M \to S_x M$.
    In contrast to $J_x(\mathbf{v}) = 0$, the equation $J_x (\mathbf{v}) \Psi(y)^* = 0$ involves only an operator of rank
    at most $2n$.

    \begin{theorem}
        \label{thm:probing}
        Let $(\H, \F, \rho)$ be a minimizing causal fermion system, $x \in M$ and $\mathbf{v} \in T_x \F$ a perturbation, then the following
        statements are equivalent
        \begin{enumerate}
            \item $\displaystyle J_x (\mathbf{v}) = 0$
            \item $\displaystyle J_x (\mathbf{v}) \Psi(y)^* = 0$ for all $y \in M$
        \end{enumerate}
    \end{theorem}

    Before we prove this proposition, we first show the following lemma which is then used in the proof.

    \begin{lemma}
        \label{lem:existence-of-a-physical-wavefunction}
        Let $(\H, \F, \rho)$ be a minimizing causal fermion system, then for every $0 \neq \ket{u} \in \H$ there exists a $x \in M$,
        such that $\spinket{\psi^u(x)} = \Psi(x) \ket{u} \neq 0$.
    \end{lemma}

    \begin{proof}
        Since $(\H, \F, \rho)$ is a minimizing causal fermion system, by Proposition~\ref{prop:el-equation} $\ell|_M = 0$ and minimal.
        Suppose, there exists a $\ket{u} \in \H$ such that for all $x \in M$, $x \ket{u} = 0$.
        Then, the operator $y = \ket{u} \bra{u} \in \F$ and $yx = xy = 0$ for all $x \in M$.
        Thus, $\ell(y) = -\mathfrak{r} \braket{u} - \mathfrak{s} < 0$.
        This contradicts the assumption, that $\ell|M = 0$ is minimal.
        Hence $(\H, \F, \rho)$ cannot be a minimizing causal fermion system if there exists a $\ket{u} \in \H$ such that $\Psi(x) \ket{u} = 0$ for all $x \in M$.
    \end{proof}

    This lemma not only shows that there is no trivial physical wavefunction, but also that union of all spin spaces spans
    a dense subspace of $\H$.
    This is the key point for the proof of Theorem~\ref{thm:probing}.

    \begin{corollary}
        \label{cor:span-of-spin-spaces-is-dense}
        Let $(\H, \F, \rho)$ be a minimizing causal fermion system, then
        \begin{align}
            \Span\left( \bigcup_{x \in M} S_x M \right) \text{ is dense in } \H.
        \end{align}
    \end{corollary}
    \begin{proof}
        Let $W \defeq \overline{\Span\left( \bigcup_{x \in M} S_x M \right)}$ and $\ket{u} \in W^\perp$.
        Suppose, $\ket{u} \neq 0$, then by Lemma~\ref{lem:existence-of-a-physical-wavefunction} there exists a $y \in M$
        such that $\pi_y \ket{u}  = \Psi(y) \ket{u} \neq 0$.
        However, this contradicts the assumption that $\ket{u} \in W^\perp$.
        Thus, we have $W^\perp = \{0\}$ and therefore, $W = \H$.
    \end{proof}

    With this corollary, we can now prove Theorem~\ref{thm:probing}.

    \begin{proof}[Proof of Theorem~\ref{thm:probing}]
        ``$(1) \Rightarrow (2)$'': This implication is trivially satisfied.

        \noindent
        ``$(2) \Rightarrow (1)$'': By definition of $\Psi(y)^*$, we have $\Img(\Psi(y)^*) = S_y M$.
        Thus, the equation $J_x (\mathbf{v}) \Psi(y)^* = 0$ for all $y \in M$ implies that
        $J_x (\mathbf{v})$ must be zero on
        \begin{align}
            \bigcup_{y \in M} \Img(\Psi(y)^*) = \bigcup_{y \in M} S_y M.
        \end{align}
        By linearity, this implies that $J_x (\mathbf{v})$ is zero on $\Span \left( \bigcup_{y \in M} S_y M \right)$,
        which by Corollary~\ref{cor:span-of-spin-spaces-is-dense} is dense in $\H$.
        Thus, $J_x (\mathbf{v}) = 0$ on all of $\H$.
    \end{proof}

    This allows us to solve $J_x (\mathbf{v}) \Psi(y)^* = 0$ for all $y \in M$ instead of solving $J_x(\mathbf{v}) = 0$.
    In particular, we can rewrite $J_x (\mathbf{v}) \Psi(y)^*$ as
    \begin{align}
        J_x (\mathbf{v}) \Psi(y)^* &= \left((D_\mathbf{v} Q) \Psi\right)(x) \Psi(y)^* + (Q D_\mathbf{v} \Psi)(x) \Psi(y)^* \nonumber \\
        &~~~- \mathfrak{r} D_\mathbf{v} \Psi(x) \Psi(y)^*\\
        &= \int_\F D_\mathbf{v} Q(x, z) P(z, y) + Q(x, z) D_\mathbf{v} P(z, y) \dd \rho(z) \nonumber \\
        &~~~- \mathfrak{r} D_\mathbf{v} P(x, y) - \left( Q \Psi(x) - \mathfrak{r} \Psi(x) \right) D_\mathbf{v} \Psi(y).
    \end{align}
    Note, that the last term in brackets vanishes due to the \ac{el} equations for the unperturbed system.
    In this way, we can express $J_x(\mathbf{v}) \Psi(y)^*$ in terms of the fermionic projector $P(x, y)$ and its perturbation $D_\mathbf{v} P(x, y)$
    along $\mathbf{v}$.

    In addition, in case of a manifold structure of $M$, we can compute derivatives of $\Psi(y)^*$ to derive tensorial equations.
    This is a key observation to relate to the equations of motion of the Standard Model.

    \subsection{Causal Fermion Systems with Manifold Structure}\label{subsec:cfs-with-manifold-structure}

    If we assume, that $M$ has a smooth manifold structure then we can use a Taylor expansion of $\Psi(y)^*$ around $x$.
    As first given in~\cite{fockbosonic} the definition of a smooth manifold structure is as follows.

    \begin{definition}
        \label{def:smooth-manifold-structure}
        The spacetime $M$ of a causal fermion system $(\H, \F_n, \rho)$ has a \textbf{smooth manifold structure} if the following conditions hold:
        \begin{itemize}
            \item[(i)] $M$ is a $k$-dimensional smooth, oriented and connected submanifold of $\F$.
            \item[(ii)] In a chart $(\phi, U)$ of $M$, the measure $\rho$ is absolutely continuous with respect to the Lebesgue
            measure with a smooth, strictly positive weight function $h$:
            \begin{align}
                \left( \phi_* \dd \rho \right)(x) = h(x) \dd^k \lambda(\phi(x)) && \text{ with } h \in C^\infty(M, \R^+).
            \end{align}
        \end{itemize}
        If in addition, the manifold $M$ is topologically of the form $M = \R \times N$
        with a manifold $N$ which admits a complete Riemannian metric $g_{N}$,
        then spacetime is said to admit a \textbf{global time function}.
    \end{definition}

    For this paper, only property is (i) of Definition~\ref{def:smooth-manifold-structure} is relevant for the following
    considerations.
    In this case, we can use the Taylor expansion of $\Psi(y)^*$ around $x$ in a chart $(\phi, U)$ of $M$.
    However, since in the abstract setting we do not have any information about analyticity of $\Psi(y)^*$, we can only
    get the implication in one direction.

    \begin{proposition}
        If $M = \supp \rho$ has a smooth manifold structure, then the equation $J_x(\mathbf{v}) = 0$ implies that
        \begin{align}
            J_x(\mathbf{v}) \partial_{\mu_1} \cdots \partial_{\mu_k} \Psi(x)^* = 0 \text{ for all $\mu_i = 1, \dots \dim M$ and $k \in \mathbb{N}_0$}.
        \end{align}
    \end{proposition}

    This direction is of course rather trivial.
    Nevertheless, from a physical point of view, this gives a systematic way of deducing tensorial equations from the linearized field equations~\eqref{eq:linearized-field-equations}.
    We also point out that the statement of the proposition is chart independent, because
    all derivatives have to vanish separately (this means that all the
    lower order terms vanish in view of the previous equations).
    As we shall see in Section~\ref{sec:04-iepsilon-regularlized-minkowski-spacetime}, the other direction also holds in the case of the $i\varepsilon$-regularized Minkowski spacetime.

    This completes the consideration in the abstract setting.
    In the next two sections, we first outline the construction of the $i \varepsilon$-regularized Minkowski spacetime
    and of physical relevant perturbations of it.
    Sequentially, we directly compute the CFS current for the continuum limit and show how we recover the Dirac Maxwell equations
    therein.

    \section{The $i\varepsilon$-Regularized Minkowski Spacetime}\label{sec:04-iepsilon-regularlized-minkowski-spacetime}

    In this section, we want to calculate the CFS current explicitly for certain perturbation $\mathbf{v}$.
    Therefore, we have to choose a realization for the Hilbert space $\H$.
    To relate the results to the standard formulation of for example electro dynamics,
    we choose $\H$ as the completion of the space of negative energy solutions of the Dirac equation
    \begin{align}
    (\slashed{k} - m)
        u(k) = 0 && \text{ with } && k_0 \leq 0  \label{eq:dirac-equation}.
    \end{align}
    By choosing this specific realization, the resulting system is assumed to have a background structure which
    corresponds to a Minkowski spacetime (with signature $+ - - -$) with no curvature.
    From the point of view of the Standard Model, this might seem an obvious assumption based on the form of the Dirac
    Equation~\eqref{eq:dirac-equation}.
    However, for causal fermion systems, this is a much more fundamental assumption, because as we see below,
    the induced measure from the Lebesgue measure of Minkowski spacetime is again continuous.
    For the description of a quantum spacetime, we rather expect a discrete measure.
    Hence, this realization can only be seen as an approximate description similar to the Standard Model and general
    relativity.

    The main problem, which arises from the assumption of continuity, is that one runs into UV divergences.
    To circumvent this problem, one introduces an artificial minimal length scale $\varepsilon > 0$ and
    a regularization operator\footnote{
        For this paper $\varepsilon$ is just an artificial parameter.
        However, as shown in~\cite[Section 4.9]{cfs} this parameter appears in the gravitational coupling constant.
    }.
    For this paper, we consider the $i\varepsilon$-regularization.

    \begin{definition}[$i\varepsilon$-regularization]
        \label{def:i-epsilon-regularization}
        Let $\varepsilon > 0$, $\ket{u} \in \H$, then the \textbf{$i\varepsilon$-regularized} wave evaluation operator is defined by
        \begin{align}
            \Psi(x) \ket{u} \defeq \int \frac{\dd^4 k}{(2 \pi)^4} u(k) e^{\frac{1}{2} \varepsilon k_0} e^{-ikx} \qquad \text{ for } x \in \R^{1, 3},
        \end{align}
        where $u(k)$ is the representation of $\ket{u}$ as a solution of~\eqref{eq:dirac-equation}.
    \end{definition}

    Note that in the above definition, $u(k)$ is a distribution supported on the lower mass shell.
    Thus, the zeroth component of the momentum~$k_0 = - \sqrt {m^2 + \vec{k}^2}$ is always less than 0.
    Hence, the exponential factor always suppresses the high energy contributions of the integral.

    For the solutions of the Dirac equation, there is an inner product defined by $u^*(x) v(x)$,
    where $u(x)^* = u^\dagger(x) \gamma^0$ is the usual Dirac adjoint of $u$.
    This is a natural choice of the causal fermion system spin inner product.
    For $x \in \R^{1, 3}$ we set the spin inner product to be
    \begin{align}
        \spinbraket{u}{v}_x &= -\bra{u} \hat{x} \ket{v} = \bra{u} \Psi(x)^* \Psi(x) \ket{v} \nonumber \\
        &\defeq \left( \Psi(x) \ket{u} \right)^\dagger \gamma^0 \left( \Psi(x) \ket{v} \right) && \text{ for all } u, v \in \H \label{eq:minkowksi-spin-inner-product}.
    \end{align}
    The signature of $\spinbraket{\cdot}{\cdot}_x$ implies that $\hat{x} = -\Psi(x)^* \Psi(x)$ has two positive and two negative eigenvalues.
    Thus, $\hat{x} \in \F$ for $n = 2$.
    Furthermore, the mapping $x \mapsto \hat{x}$ induces a measure $\rho^\varepsilon$ given by the push-forward of the
    Lebesgue measure on $\R^{1, 3}$.
    We refer to the causal fermion system $(\mathcal{H}, \mathcal{F}, \rho^\varepsilon)$ as the regularized Minkowski vacuum.
    The operator $\hat{x}$ is regular, i.e.~it has exactly two positive and two negative eigenvalues.
    Thus, the spin spaces at each point $S_{x} M \isoeq \C^4$.

    The adjoint of $i\varepsilon$-regularized wave evaluation operator $\Psi$ with respect to the spin inner product
    defined by~\eqref{eq:minkowksi-spin-inner-product} is given by
    \begin{align}
        \label{eq:def-i-epsilon-regularized-wave-evaluation-adjoint}
        \Psi(x)^* \phi = \left(k \mapsto (\slashed{k} + m) \delta(k^2 - m^2) \Theta(-k_0) e^{\frac{1}{2} \varepsilon k_0} e^{ikx} \phi \right) && \text{ with } && \phi \in \C^4 \isoeq S_x M.
    \end{align}

    \begin{theorem}
        \label{thm:expansion}
        Let $\Psi$ be the $i\varepsilon$-regularized wave evaluation operator given by Definition~\ref{def:i-epsilon-regularization},
        then the following statements are equivalent
        \begin{enumerate}
            \item $\displaystyle J_x (\mathbf{v}) = 0$
            \item $\displaystyle J_x (\mathbf{v}) \partial_{\mu_1} \dots \partial_{\mu_n} \Psi(x)^* = 0$ for all $n \in \N_0$ and $\mu_i \in \{0, 1, 2, 3\}$
        \end{enumerate}
    \end{theorem}
    \begin{proof}
        ``$(1) \Rightarrow (2)$'': This implication is trivially satisfied.

        \noindent
        ``$(2) \Rightarrow (1)$'': Let $\phi \in \C^4$, then $g(y) \defeq J_x(\mathbf{v}) \Psi(y)^* \phi$ defines a smooth function on $\R^{1, 3}$.
        By the Paley-Wiener theorem for the Fourier transform in the complex domain
        (see for example~\cite[Section~VI.4]{yosida}), we know that $\Psi(z)^* \phi$ is holomorphic for $z \in B_{\frac{\varepsilon}{2}}(y) \subset \C^4$
        for all $y \in \R^{1, 3}$.
        In particular, $\Psi(y)^* \phi$ is real analytic in $B_{\frac{\varepsilon}{2}}(y)$ for any $y \in \R^{1, 3}$.
        Thus, $g$ is real analytic in $B_{\frac{\varepsilon}{2}}(y)$ and hence the condition $(2)$ implies that $g = 0$ on $B_{\frac{\varepsilon}{2}}(x)$.
        By iteratively applying the same argument and using that~$\R^{1, 3}$ is connected, we get $g = 0$ on all of $\R^{1,3}$.
        By Theorem~\ref{thm:probing} we conclude that $J_x (\mathbf{v}) = 0$.
    \end{proof}

    Using~\eqref{eq:def-kernel-of-the-fermionic-projector} for the kernel of the fermionic projector, we get
    \begin{align}
        \label{eq:minkowski-kernel-of-the-fermionic-projector}
        P(x, y) &= \int \frac{\dd k^4}{(2 \pi)^2} (\slashed{k} + m) \delta(k^2 - m^2) \Theta(-k_0) e^{-ik(x - y)} e^{\varepsilon k_0}.
    \end{align}
    In the following we denote the regularized tangent vector by $\xi^{\mu} \defeq (y^0 - x^0 - i\varepsilon, \vec{y} - \vec{x})$.
    Then, a direct computation of the integral gives
    \begin{align}
        \label{eq:minkowski-vector-scalar-fermionic-projector}
        P(x, y) = (i \slashed{\partial}^{(x)} + m) T^{(0)}_{m^2}(x, y) = \frac{i}{2} T^{(-1)}_{m^2} \slashed{\xi} + m T^{(0)}_{m^2}
    \end{align}
    with
    \begin{align}
        T^{(0)}_{m^2}(x, y) \defeq \frac{m^2}{8 \pi^3} \frac{K_1(z)}{z} && \text{ and } && T^{(-1)}_{m^2}(x, y) \defeq \frac{m^4}{8 \pi^3} \frac{K_2(z)}{z^2}.
    \end{align}
    In the above expressions $K_1$ and $K_2$ are modified Bessel functions of second kind and their argument $z$ is defined by
    \begin{align}
        \label{eq:def-bessel-function-argument}
        z(x, y) \defeq m \sqrt {-\xi_\mu \xi^\mu} = m \sqrt {\norm{\vec{y} - \vec{x}}^2 - (y^0 - x^0 - i\varepsilon)^2}.
    \end{align}
    Before we continue, we also introduce the notation
    \begin{align}
        t \defeq y^0 - x^0 && \text{ and } && r = \norm{\vec{y} - \vec{x}}^2.
    \end{align}

    Having defined $P(x, y)$ for the Minkowski vacuum, we continue with computing the CFS current for this specific
    choice of the kernel of the fermionic projector.
    In addition, we also define the perturbations of $P(x, y)$ for a particle, i.e.~a positive energy solution of the Dirac Equation~\eqref{eq:dirac-equation},
    and a vector potential in this setting.

    \subsection{The CFS Current in Minkowski Spacetime}\label{subsec:cfs-current-in-minkowski-spacetime}

    In order to compute the second variation of the Lagrangian, we have to determine the eigenvalues of the operator product
    $\hat{x} \hat{y}$.
    As stated in Section~\ref{subsec:the-causal-action-principle}, the non-vanishing eigenvalues coincide with the
    eigenvalues of the closed chain $A_{xy}$.

    \begin{proposition}
        \label{prop:eigenvalues-of-the-closed-chain}
        Let $P(x, y)$ be the kernel of the fermionic projector of the Minkowski vacuum, then the eigenvalues of the closed chain
        $A_{xy} = P(x, y) P(y, x)$ are given by
        \begin{align}
            \label{eq:eigenvalues-of-the-closed-chain}
            \lambda^{xy}_{\pm} = b \pm \sqrt {\frac{1}{4} \tr_{S_x M} \left[ \left( A_{xy} - b \right)^2 \right]}
        \end{align}
        each with a two-fold degeneracy and $b \defeq \frac{1}{4}\tr_{S_x M} A_{xy}$.
    \end{proposition}
    \begin{proof}
        Using~\eqref{eq:minkowski-vector-scalar-fermionic-projector} to compute the closed chain gives
        \begin{align*}
            A_{xy} = b + a_{\mu} \gamma^{\mu} + c_{\mu\nu} \Sigma^{\mu\nu}
        \end{align*}
        with the real valued functions
        \begin{align*}
            b &= \frac{1}{4} \abs{T^{(-1)}_{m^2}}^2 \xi_\mu \complexconj{\xi^\mu} + m^2 \abs{T^{(0)}_{m^2}}^2, \\
            a_{\mu} &= \frac{i}{2} m \xi_{\mu} T^{(-1)}_{m^2} \complexconj{T^{(0)}_{m^2}} - \frac{i}{2} m \complexconj{\xi_{\mu} T^{(-1)}_{m^2}} T^{(0)}_{m^2}, \\
            c_{\mu\nu} &= \frac{1}{8i} \abs{T^{(-1)}_{m^2}}^2 \left( \xi_\mu \complexconj{\xi_\nu} - \complexconj{\xi_\mu} \xi_\nu \right)
        \end{align*}
        and $\Sigma^{\mu\nu} = \frac{i}{2} \left[ \gamma^\mu, \gamma^\nu \right]$.
        The $\gamma$ and the $\Sigma$ matrices are trace-free, thus
        \begin{align*}
            b &= \frac{1}{4}\tr_{S_x M} A_{xy} & \text{ and } && \left( A_{xy} - b \right)^2 &= a_{\mu} a^{\mu} + 2 c_{\mu\nu} c^{\mu \nu}.
        \end{align*}
        From the above expression, we directly see that the eigenvalues of $A_{xy}$ are given by
        \begin{align*}
            \lambda^{xy}_{\pm} = b \pm \sqrt {a_{\mu} a^{\mu} + 2 c_{\mu\nu} c^{\mu \nu}}  = b \pm \sqrt {\frac{1}{4} \tr_{S_x M} \left[ \left( A_{xy} - b \right)^2 \right]}
        \end{align*}
        each with a two-fold degeneracy.
    \end{proof}

    This proposition is very helpful for computing the variation of the eigenvalues in terms of the variation of the closed chain
    as we have the explicit dependency on $A_{xy}$.

    Before we continue with the computation, we want to highlight that for the proof we represented
    the closed chain $A_{xy}$ — a symmetric operator on the spin space $S_x M$ — as a linear combination of $(\mathbb{1}, \gamma^\mu, \Sigma^{\mu \nu})$.
    However, when analyzing the components $c_{\mu\nu}$ of the bi-linear term we see that
    \begin{align}
        c_{00} &= c_{ij} = 0 &\text{ and } && c_{0i} &= - \frac{1}{4} \abs{T^{(-1)}_{m^2}}^2 \varepsilon \xi_j.
    \end{align}
    Thus, it is more convenient to work with the linear independent matrices $\Gamma^i \defeq \Sigma^{0i}$ and $i \gamma^5 \Gamma^i$
    instead.
    Correspondingly, we introduce $c_i \defeq 2 c_{0i}$.
    All 16 basis elements of the symmetric spin operators are then given by $(\mathbb{1}, i\gamma^5, \gamma^\mu, \gamma^5 \gamma^\mu, \Gamma^i, i \gamma^5 \Gamma^i)$.
    We refer to the components as scalar, pseudo-scalar, vector, pseudo-vector, bi-linear and pseudo-bi-linear respectively.
    In this basis, we have
    \begin{align}
        P(x, y) = \frac{i}{2} T^{(-1)}_{m^2} \xi_{\mu} \gamma^{\mu} + m T^{(0)}_{m^2} && \text{ and } && A_{xy} = b + a_{\mu} \gamma^{\mu} + c_i \Gamma^i.
    \end{align}

    \begin{proposition}
        \label{prop:closed-chain-diagonalizable}
        Let $x, y \in \R^{1, 3}$, then $A_{xy}$ is diagonalizable if
        \begin{align}
            \frac{1}{4} \tr_{S_x M} \left[ \left( A_{xy} - b \right)^2 \right] \neq 0.
        \end{align}
        In this case
        \begin{align}
            \Lambda^{xy}_{\pm} \defeq \frac{1}{2} \pm \frac{A_{xy} - b}{\sqrt {\tr_{S_x M} \left[ \left( A_{xy} - b \right)^2 \right]}}
        \end{align}
        define the projectors onto the eigenspaces for the corresponding eigenvalues $\lambda^{xy}_{\pm}$.
    \end{proposition}
    \begin{proof}
        The first statement follows directly from the definition of light-like separation and
        the Proposition~\ref{prop:eigenvalues-of-the-closed-chain} for the eigenvalues.

        For the second statement, we use that $\delta \lambda^{xy}_i = \tr\left[ \Lambda^{xy}_i \delta A_{xy} \right]$ and thus
        \begin{align*}
            \tr\left[ \Lambda^{xy}_\pm \delta A_{xy} \right] = 2 \delta \lambda^{xy}_{\pm}
            = \frac{1}{2} \tr \left[ \left( 1 \pm \frac{A_{xy} - b}{\sqrt {\frac{1}{4} \tr_{S_x M} \left[ \left( A_{xy} - b \right)^2 \right] }} \right) \delta A_{xy} \right].
        \end{align*}
    \end{proof}

    Note that the projectors $\Lambda^{xy}_{\pm}$ have rank $2$.
    Hence, they do not coincide with the eigenvalue projectors from Proposition~\ref{prop:abstract-setting-q-and-symmetry-property}.

    In Section~\ref{sec:the-continuum-limit-of-minkowski-spacetime}, we directly give the eigenspace projectors in the continuum
    limit.
    Before analyzing the continuum limit, we want to determine the physically interesting perturbations of the Minkowski vacuum.

    \subsection{Dirac Current}\label{subsec:dirac-current}

    As a first example for a CFS current, we want to analyze the case of perturbing the vacuum minimizer by
    additional particles with positive frequency.
    This is archieved by a tangent vector $\mathbf{u_p}$ defined by
    \begin{align}
        D_{\mathbf{u_p}} P(x, y) = \frac{1}{2 \pi} \sum_{i = 1}^{n_p} \phi_i(x) \diracad{\phi_i(y)},
    \end{align}
    where $n_p$ is the particle number and $\phi_i$ are solutions to the Dirac equation with positive energy.
    Similarly, one can introduce anti-particles by creating holes in the Dirac sea.
    This corresponds to the tangent vector $\mathbf{u_a}$ defined by
    \begin{align}
        D_{\mathbf{u_a}} P(x, y) = -\frac{1}{2 \pi} \sum_{i = 1}^{n_a} \psi_i(x) \diracad{\psi_i(y)},
    \end{align}
    where again $n_a$ is the number of anti-particles.
    $\psi_i$ are negative energy solutions of the Dirac equation.

    \subsection{Maxwell Current}\label{subsec:maxwell-current}

    In order to derive the corresponding perturbation related to a vector potential, we look at a collective perturbation
    of the kernel of the fermionic projector.
    Instead of the vacuum Dirac equation, we analyze the Dirac equation with a potential $A(x)$
    \begin{align}
        \label{eq:dirac-maxwell-equation}
        (i \slashed{\partial} - \slashed{A}(x) - m) u(x) = 0.
    \end{align}
    At this point, we do not consider $A$ as an electromagentic potential but rather as a small perturbation.
    Thus, we take the ansatz that the perturbed fermionic projector is a solution to the perturbed Dirac Equation~\eqref{eq:dirac-maxwell-equation}.
    Hence,
    \begin{align}
    (i \slashed{\partial}^{(x)} - \slashed{A}(x) - m)
        (P(x, y) + D_{\mathbf{v}} P(x, y)) = 0.
    \end{align}
    We solve this equation iteratively in orders of $A$:
    \begin{align}
        P^{(0)}(x, y) &\defeq P(x, y), \\
        (i \slashed{\partial}^{(x)} - m) P^{(n)}(x, y) &= \slashed{A}(x) P^{(n - 1)}(x, y),
    \end{align}
    then the full perturbation is given by
    \begin{align}
        D_{\mathbf{v}} P(x, y) &= \sum_{n = 1}^{\infty} P^{(n)}(x, y).
    \end{align}

    For this paper, we only consider the first order in $A$, i.e.~$D_{\mathbf{v}} P(x, y) = P^{(1)}(x, y) + \mathcal{O}(A^2)$.
    As shown in Appendix~\ref{sec:the-regularized-light-cone-expansion} $P^{(1)}(x, y)$ can be written as
    \begin{align}
        \label{eq:maxwel-pertubration-derivation}
        P^{(1)}(x, y) = (i \slashed{\partial}^{(x)} + m) \gamma^\mu (-i \slashed{\partial}^{(y)} + m) \Delta T_{m^2}[A_\mu](x, y),
    \end{align}
    where $\Delta T$ is a solution of
    \begin{align}
    (\dAlembertian + m^2)
        \Delta T_{m^2}[A](x, y) = A(x) T^{(0)}_{m^2}(x, y) \label{eq:klein-gordon-equation-for-delta-t}.
    \end{align}

    \begin{proposition}
        \label{prop:electro-magnetic-solution}
        The formal series
        \begin{align}
            \Delta T_{m^2}[A](x, y) \defeq \sum_{n = 0}^{\infty} \frac{T_{m^2}^{(n + 1)}(x, y)}{n!} \int_0^1 (\tau - \tau^2)^n (\dAlembertian^n A)(x + \tau \xi) \dd \tau
        \end{align}
        with $\xi^\mu = (y^0 - x^0 - i\varepsilon, \vec{y} - \vec{x})$ and $T_a^{(n)} \defeq \pdv[n]{a} T_a$ solves~\eqref{eq:klein-gordon-equation-for-delta-t}.
    \end{proposition}

    The proof of this proposition is given in Appendix~\ref{sec:the-regularized-light-cone-expansion}.
    With this result, we compute $D_{\mathbf{v}} P(x, y)$ according to~\eqref{eq:maxwel-pertubration-derivation}.
    In the continuum limit the relevant term, which contributes to the CFS current is given by
    \begin{align}
        D_{\mathbf{v}} P(x, y) = - \frac{1}{2} \slashed{\xi} \xi^{\mu} \int_0^1 \partial^\nu F_{\nu\mu}(x + \tau \xi) \dd \tau T^{(0)}_{m^2}(x, y),
    \end{align}
    where $F_{\mu\nu} (x) \defeq \partial_\mu A_\nu (x) - \partial_\nu A_\mu(x)$ is the usual field tensor associated to
    the potential $A$.
    Thus, $\partial^\nu F_{\nu\mu}(x)$ is the usual current corresponding to $A$.

    \section{The Continuum Limit of Minkowski Spacetime}\label{sec:the-continuum-limit-of-minkowski-spacetime}

    In this section, we analyze the behaviour around $\varepsilon \approx 0$.
    We refer to the limit $\varepsilon \to 0$ as continuum limit.
    However, the kernel of the fermionic projector is not regular on the light-cone for $\varepsilon = 0$.
    In~\cite{cfs} this limit is analyzed as a distribution.
    In this paper, we present an alternative approach.
    The advantages of this approach are outlined in Section~\ref{sec:conclusion-and-outlook}.

    \subsection{Analysis of the UV Divergence}\label{subsec:analysis-of-the-uv-divergence}

    In the following we consider a small $\varepsilon > 0$.
    Then we start by rewriting the kernel of the fermionic projectors given by~\eqref{eq:minkowski-vector-scalar-fermionic-projector} as
    \begin{align}
        P(x, y) = \frac{i}{2} T^{(-1)}_{m^2} \slashed{\xi} + m T^{(0)}_{m^2} = T^{(-1)}_{m^2} \left( \frac{i}{2} \slashed{\xi} + m \frac{T^{(0)}_{m^2}}{T^{(-1)}_{m^2}} \right).
    \end{align}
    The factor $T^{(-1)}_{m^2}$ diverges on the light-cone for $\varepsilon \to 0$ where as the factor $\frac{T^{(0)}_{m^2}}{T^{(-1)}_{m^2}}$ is regular.
    Thus, the factor $T^{(-1)}_{m^2}$ describes the singular structure of the fermionic projectors in the continuum limit.
    This property is displayed in Figure~\ref{fig:absolut-value-of-t-minus-one-squared} where one sees $\abs{T^{(-1)}_{m^2}}$
    has its largest value at the origin and decays rapidly away from the origin.
    In particular, one also sees slower decay along the light-cone.

    \begin{figure}
        \centering

        \begin{tikzpicture}
            \begin{axis}
                [
                view={0}{90},
                xlabel=$r$,
                ylabel=$t$,
                xtick={-3, 0, 3},
                xticklabels={$-3 \varepsilon$, $0$, $3 \varepsilon$},
                ytick={-3, 0, 3},
                yticklabels={$-3 \varepsilon$, $0$, $3 \varepsilon$},
                title={Plot of $\abs{T^{(-1)}_{m^2}(t, r)}^2$},
                colorbar,
                colorbar style={
                    ytick={0, 6000},
                    yticklabels={0, $\frac{1}{4 \pi^6 \varepsilon^8}$}
                }
                ]

                \addplot3[surf, mesh/rows=50, shader=interp] table {absolut-value-of-alpha-squared.dat};
                \addplot3[gray, thick] table {absolut-value-of-alpha-squared-contour-2.dat};
                \addplot3[gray, thick] table {absolut-value-of-alpha-squared-contour-1.dat};
                \addplot3[gray, thick] table {absolut-value-of-alpha-squared-contour-0.dat};
            \end{axis}
        \end{tikzpicture}

        \caption{
            The plot shows $\abs{T^{(-1)}_{m^2}(t, r)}^2$ for $\varepsilon > 0$ and its contour lines for at $\abs{T^{(-1)}_{m^2}(0, 0)}^2 \cdot 10^{-1}$, $\abs{T^{(-1)}_{m^2}(0, 0)}^2  \cdot 10^{-2}$ and $\abs{T^{(-1)}_{m^2}(0, 0)}^2 \cdot 10^{-3}$.
            One clearly sees a rapid decay of $\abs{T^{(-1)}_{m^2}}^2$ away from the origin with slower decay on the light-cone.
        }
        \label{fig:absolut-value-of-t-minus-one-squared}
    \end{figure}
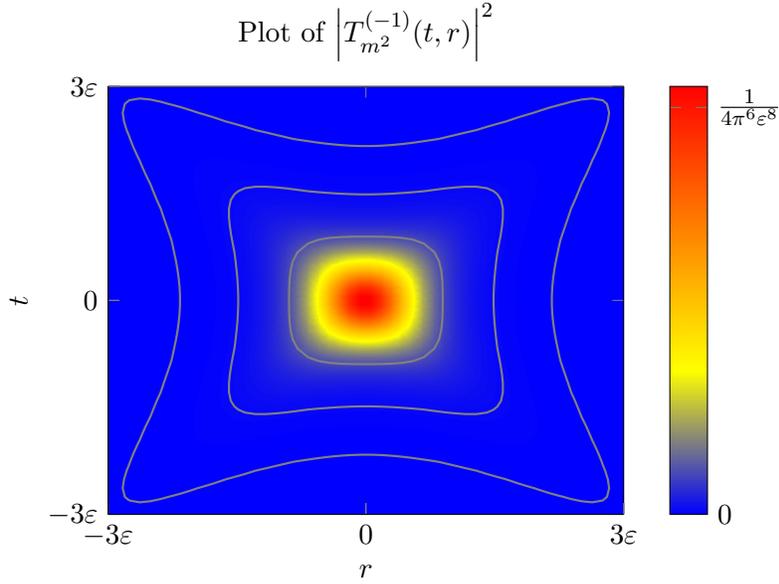

    Thus, for the following, we consider $x_\mu - y_\mu$ of order $\varepsilon$ with $t^2 - r^2$ of order $\varepsilon^2$.
    Then
    \begin{align}
        \frac{T^{(0)}_{m^2}}{T^{(-1)}_{m^2}} = \frac{1}{2} \left( r^2 - t^2 + \varepsilon^2 + 2 i \varepsilon t \right) + \mathcal{O}(\varepsilon^3).
    \end{align}
    In addition, the components of the closed chain are given by
    \begin{align}
        b &= \frac{\abs{T^{(-1)}_{m^2}}^2}{4} (t^2 - r^2 + \varepsilon^2 + \mathcal{O}(\varepsilon^4)), \\
        a_{\mu} &= \frac{\abs{T^{(-1)}_{m^2}}^2}{4} \left( 2 m \varepsilon t (y_\mu - x_\mu) + \mathcal{O}(\varepsilon^4) \right) = \abs{T^{(-1)}_{m^2}}^2 \mathcal{O}(\varepsilon^3), \\
        c_i &= - \frac{\abs{T^{(-1)}_{m^2}}^2}{2} \varepsilon (y_i - x_i).
    \end{align}
    Therefore, the eigenvalues are given by
    \begin{align}
        \label{eq:eigenvalues-in-the-continuum}
        \lambda_{\pm} = \frac{\abs{T^{(-1)}_{m^2}}^2}{4} \left( t^2 - r^2 + \varepsilon^2 \pm 2 i \varepsilon r + \mathcal{O}(\varepsilon^3) \right).
    \end{align}
    For $r > 0$ the eigenvalues form a complex conjugated pair, hence
    $\abs{\lambda_+} = \abs{\lambda_-}$ and by Proposition~\ref{prop:closed-chain-diagonalizable} the closed chain is diagonalizable.

    In the definition of the CFS current, we integrate over all $t$ and $r$.
    Since the contribution for $r = 0$ has measure zero, we restrict to the case for $r > 0$ and drop terms of order $\varepsilon^3$
    from now on.
    Considering the components and the eigenvalues of the closed chain as above, we get the following projection operators
    of the corresponding eigenspaces
    \begin{align}
        \Lambda_{\pm} = \frac{1}{2} \pm i \frac{(y_i - x_i)}{2r} \Gamma^i
    \end{align}
    which commute with the chiral projectors $\chi_{L/R} = \frac{1}{2} \mp \frac{1}{2} \gamma^5$.
    Thus, we get the spectral operators $\Lambda_i = \chi_{L/R} \Lambda_\pm$.
    In this notation $i = 1, 2, 3, 4$ correspond to the projectors $(L, +), (L, -), (R, +), (R, -)$ respectively.

    Since the eigenvalues of the closed chain form complex conjugated pairs, all eigenvalues have the same absolute
    value.
    Hence, $\abs{\lambda^{xy}_i} - \abs{\lambda^{xy}_j} = 0$ for all $i$ and $j$.
    Thus, in the continuum limit, $Q(x, y) = 0$.
    For the Lagrange multiplier $\mathfrak{r} = 0$, the $i\varepsilon$-regularized Minkowski vacuum solves the
    restricted \ac{el}~\eqref{eq:el-equation}.

    Sequentially, in the continuum limit the CFS current reduces to
    \begin{align}
        \label{eq:cfs-current-in-the-continuum-limit}
        J_x(\mathbf{v}) &= \int D_\mathbf{v} Q(x, y) \Psi(y) \dd^4 y \nonumber \\
        &= \frac{1}{4} \int \sum_{i, j = 1}^{4} D_\mathbf{v} \left( \abs{\lambda^{xy}_i} - \abs{\lambda^{xy}_j} \right) \frac{\complexconj{\lambda^{xy}_i}}{\abs{\lambda^{xy}_i}} \Lambda^{xy}_i P(x, y) \Psi(y) \dd^4 y \nonumber \\
        &= \frac{1}{4} \int \sum_{i, j = 1}^{4} \Re \tr_{S_x M} \left[ \Omega^{xy}_{ij} D_\mathbf{v} A_{xy} \right] \frac{\complexconj{\lambda^{xy}_i}}{\abs{\lambda^{xy}_i}} \Lambda^{xy}_i P(x, y) \Psi(y) \dd^4 y.
    \end{align}
    with $\Omega^{xy}_{ij} \defeq \frac{\complexconj{\lambda^{xy}_i}}{\abs{\lambda^{xy}_i}} \Lambda^{xy}_{i} - \frac{\complexconj{\lambda^{xy}_j}}{\abs{\lambda^{xy}_j}} \Lambda^{xy}_{j}$.

    The remaining task is to determine the term $\Re \tr_{S_x M} \left[ \Omega^{xy}_{ij} D_\mathbf{v} A_{xy} \right]$ for
    different perturbation $\mathbf{v}$.
    In the following, we consider the perturbations described in Section~\ref{subsec:dirac-current} and Section~\ref{subsec:maxwell-current}.

    \subsection{Axial Current for one Fermion Type}\label{subsec:axial-current-for-one-fermion-type}

    By the structure of the Lagrangian one sees that for a single fermion type ($n = 2$) we only get a non-vanishing causal fermion system
    current if the eigenvalues are perturbed differently.
    This is the case of axial currents.

    We start by analyzing the Dirac current of a single particle.
    As stated in the previous section, the corresponding perturbation of the kernel of the fermionic projector is given by
    \begin{align*}
        D_{\mathbf{u}} P(x, y) = \frac{1}{2 \pi} \phi(x) \diracad{\phi(y)}.
    \end{align*}
    By the assumptions made in Section~\ref{subsec:analysis-of-the-uv-divergence}, we analyze the regime at the order of $\varepsilon$.
    Physically, the regime corresponds to the Planck scale.
    Thus, on this scale we assume that $\phi$ is almost constant.
    Hence,
    \begin{align}
        \label{eq:localization-approximation}
        D_{\mathbf{u}} P(x, y) \approx \frac{1}{2 \pi} \phi\left( \frac{x + y}{2}\right) \diracad{\phi\left( \frac{x + y}{2}\right) }.
    \end{align}

    For this perturbation, we get the following CFS current (see Appendix~\ref{sec:computation-of-the-cfs-current} for more
    details on the calculation):
    \begin{align}
        &\text{0th Order:} & J_x(\mathbf{u}) \Psi(x)^* &= i C^{(D, 0)}_{\text{Pseudo-Bilinear}} ~ \diracad{\phi(x)} \gamma^5 \gamma_k \phi(x) ~ \gamma^5 \Gamma^k \\
        &\text{1st Order:} & J_x(\mathbf{u}) \partial_0 \Psi(x)^* &= i C^{(D, t)}_{\text{Pseudo-Scalar}} ~ \diracad{\phi(x)} \gamma^5 \gamma_0 \phi(x) ~ \gamma^5 \nonumber \\
        &&&~~~ + i C^{(D, t)}_{\text{Pseudo-Bilinear}}  ~ \diracad{\phi(x)} \gamma^5 \gamma_k \phi(x) ~ \gamma^5 \Gamma^k,\\
        & & J_x(\mathbf{u}) \partial_k \Psi(x)^* &= i C^{(D, r)}_{\text{Pseudo-Scalar}} ~ \diracad{\phi(x)} \gamma^5 \gamma_k \phi(x) ~ \gamma^5 \nonumber \\
        &&&~~~ + i C^{(D, r)}_{\text{Pseudo-Bilinear}}  ~ \diracad{\phi(x)} \gamma^5 \gamma_0 \phi(x) ~ \gamma^5 \Gamma_k.
    \end{align}
    The constants $C^{(D, 0)}_{\text{Pseudo-Bilinear}}$, $C^{(D, t)}_{\text{Pseudo-Bilinear}}$ and $C^{(D, r)}_{\text{Pseudo-Scalar}}$
    are all real and are explicitly given in Appendix~\ref{sec:computation-of-the-cfs-current}.
    However, the numerical values of these constants do not matter to derive the classical equations.
    The relevant aspect of the above expression is that axial Dirac current $\diracad{\phi(x)} \gamma^5 \gamma_{\mu} \phi(x)$
    appears in coefficients of the operator-valued CFS current.

    Taking a closer look at the expressions of the constants, one clearly sees that they are all different.
    Thus, the axial Dirac current does not appear in a Lorentz invariant way.
    The breaking of Lorentz invariance on the level of the CFS current is not surprising, because we regularized the
    vacuum along a designated time direction.
    However, we also see that considering the current generated by a Dirac particle, the only possible way of having $J_x(u) = 0$
    is that $\diracad{\phi(x)} \gamma^5 \gamma_{\mu} \phi(x) = 0$.
    In this particular case, the equations are trivially Lorentz invariant.
    In order to allow for non-trivial particle currents, we need to additionally consider an axial Maxwell potential, which
    compensates for the particle current.
    As we see below, the constraining equations which ensure that both contributions compensate each other, are then Lorentz invariant.

    Since, only the axial particle current occurs in the CFS current, we can compensate the current generated by a
    Dirac particle by a collective perturbation of the Dirac sea by a
    pseudo vector potential $A_\mu(x) \gamma^5 \gamma^\mu$.
    As mentioned in Section~\ref{subsec:maxwell-current}, the contributing part to the fermionic projector is given by
    \begin{align}
        D_{\mathbf{v}} P(x, y) &= \frac{1}{2} \gamma^5 \slashed{\xi} \xi^{\mu} \int_0^1 (\tau - \tau^2) (\partial^\nu F_{\mu \nu})(x + \tau \xi) \dd \tau ~ T^{(0)}_{m^2}(x, y) \nonumber \\
        &~~~+\gamma^5 \gamma^{\mu} \int_0^1 (1 - \tau)^2 (\partial^\nu F_{\mu \nu})(x + \tau \xi) \dd \tau ~ T^{(1)}_{m^2}(x, y) \\
        &\approx \frac{1}{12} \gamma^5 \slashed{\xi} \xi^{\mu} \partial^\nu F_{\mu \nu}\left( \frac{x + y}{2}  \right) T^{(0)}_{m^2}(x, y)
        + \frac{1}{3} \gamma^5 \gamma^{\mu} \partial^\nu F_{\mu \nu}\left( \frac{x + y}{2}  \right) T^{(1)}_{m^2}(x, y)
    \end{align}
    where $F_{\mu \nu}(x) \defeq \partial_\mu A_\nu(x) - \partial_\nu A_\mu(x)$ is the usual field strength tensor.
    Similar to the consideration above, we assumed that $A_\nu(x)$ and its derivates do not vary on Planck scales.
    The corresponding CFS current is then given by
    \begin{align}
        &\text{0th Order:} & J_x(\mathbf{u}) \Psi(x)^* &= i C^{(M, 0)}_{\text{Pseudo-Bilinear}} ~ j_k ~ \gamma^5 \Gamma^k \\
        &\text{1st Order:} & J_x(\mathbf{u}) \partial_0 \Psi(x)^* &= i C^{(M, t)}_{\text{Pseudo-Scalar}} ~ j_0 ~ \gamma^5 \nonumber \\
        &&&~~~ + i C^{(M, t)}_{\text{Pseudo-Bilinear}}  ~ j_k ~ \gamma^5 \Gamma^k,\\
        & & J_x(\mathbf{u}) \partial_k \Psi(x)^* &= i C^{(M, r)}_{\text{Pseudo-Scalar}} ~ j_k ~ \gamma^5 \nonumber \\
        &&&~~~ + i C^{(M, r)}_{\text{Pseudo-Bilinear}}  ~ j_0 ~ \gamma^5 \Gamma_k
    \end{align}
    where $j_\mu = \partial^\nu F_{\nu \mu}$.
    The explicit expressions for the constants appearing in the above equations are also given in Appendix~\ref{sec:computation-of-the-cfs-current}.
    In addition, it is shown that the ratios of the respective Dirac and Maxwell
    constants are all equal, real and in the approximation do not depend on $\varepsilon$, i.e.
    \begin{align}
        \label{eq:equality-of-the-current-constant-ratios}
        \frac{1}{\alpha} &\defeq \frac{C^{(M, 0)}_{\text{Pseudo-Bilinear}}}{C^{(D, 0)}_{\text{Pseudo-Bilinear}}} = \frac{C^{(M, t)}_{\text{Pseudo-Scalar}}}{C^{(D, t)}_{\text{Pseudo-Scalar}}}
        = \frac{C^{(M, t)}_{\text{Pseudo-Bilinear}}}{C^{(D, t)}_{\text{Pseudo-Bilinear}}} \nonumber \\
        &= \frac{C^{(M, r)}_{\text{Pseudo-Scalar}}}{C^{(D, r)}_{\text{Pseudo-Scalar}}} = \frac{C^{(M, r)}_{\text{Pseudo-Bilinear}}}{C^{(D, r)}_{\text{Pseudo-Bilinear}}}
    \end{align}
    This allows us to write down a combined perturbation $\mathbf{u} + \mathbf{v}$, which solves the linearized field~\eqref{eq:linearized-field-equations}.
    This is the case, if and only if
    \begin{align}
        j^{(a)}_\mu(x) = \alpha \diracad{\phi(x)} \gamma^5 \gamma_{\mu} \phi(x).
    \end{align}
    Note, that due to~\eqref{eq:equality-of-the-current-constant-ratios} the constraining equation is again Lorentz invariant
    and the coupling constant $\alpha$ is real-valued.

    Therefore, for a single fermion type we recover the classical equations for an axial vector potential in a Lorentz invariant.
    For one fermion type, causal fermion systems do not allow for vectorial interaction.
    This is due to the fact that a vectorial perturbation preserves the twofold degeneracy of the eigenvalues of the closed chain.
    Thus, such a perturbation does not contribute in the Lagrangian~\eqref{eq:lagrangian}.
    In the following section, we show how to recover the equation of motion for a vector perturbation by considering two fermion
    types, where one type is kept unperturbed.

    \subsection{Electromagnetic Interaction}\label{subsec:electromagnetic-interaction}

    In this section, we consider two types, $e$ and $\nu$, of fermions.
    Thus, we form the tuple $(e, \nu)$ and consider the causal fermion system of spin-dimension $n = 4$.
    Then the continuum limit vacuum kernel of the fermionic projector is given by
    \begin{align}
        P(x, y) = \frac{i}{2} T^{(-1)}_{m^2} \xi_{\mu} (\gamma^{\mu} \oplus \gamma^{\mu}) = \frac{i}{2} T^{(-1)}_{m^2} \xi_{\mu} \begin{pmatrix}
                                                                                                                                     \gamma^{\mu} & 0            \\
                                                                                                                                     0            & \gamma^{\mu}
        \end{pmatrix}
    \end{align}
    In this setup, we get a non-vanishing scalar and bilinear contribution to the CFS current from a collective vector
    perturbation if we only couple it to one of the two fermion types.
    For this, we choose the label $e$ as the fermion type, which Dirac sea is perturbed.

    In the language of the Standard Model, the vector perturbation corresponds to the electromagnetic field, i.e.~to
    photons.
    The fermion type coupling to this perturbation is a charged fermion, which in our setup is the electron, whereas the
    uncharged fermion type is the neutrino.

    By only considering a vector perturbation, our model only allows for a creation of $e$-particle, because there is no
    contribution canceling the CFS current for particles of type $\nu$.
    This requires additional, axial perturbations analogously to Section~\ref{subsec:axial-current-for-one-fermion-type}.

    First, the CFS current induced by an $e$-particle $\phi$ corresponding to a perturbation $\mathbf{u_p}$ is given by
    \begin{align}
        J_x(\mathbf{u_p}) \Psi(x)^* &= 2 C^{(D, 0)}_{\text{Scalar}} \diracad{\phi(x)} \gamma_0 \phi(x) \nonumber \\
        &~~~+ \frac{5i}{2} C^{(D, 0)}_{\text{Pseudo-Bilinear}} \diracad{\phi(x)} \gamma^5 \gamma_k \phi(x) \gamma^5 \Gamma^k.
    \end{align}
    Note, that the second term is, up to the factor $\frac{5}{2}$, the same as the contribution from the axial current.
    This also holds for the higher moments.
    Thus, we do not repeat these contributions here.

    For the first moment of the CFS current, we have
    \begin{align}
        J_x(\mathbf{u_p}) \partial_0 \Psi(x)^* &\containsContribution i C^{(D, t)}_{\text{Scalar}} \diracad{\phi(x)} \gamma_0 \phi(x) + i C^{(D, t)}_{\text{Bilinear}} \diracad{\phi(x)} \gamma^k \phi(x) \Gamma_k,\\
        J_x(\mathbf{u_p}) \partial_k \Psi(x)^* &\containsContribution i C^{(D, r)}_{\text{Scalar}} \diracad{\phi(x)} \gamma_k \phi(x) + i C^{(D, r)}_{\text{Bilinear}} \diracad{\phi(x)} \gamma^0 \phi(x) \Gamma_k.
    \end{align}

    Similarly, the contribution from an electro magnetic perturbation $\mathbf{u_{A}}$ is given by
    \begin{align}
        J_x(\mathbf{u_A}) \Psi(x)^* = \tilde{C}^{(M, 0)}_\text{Scalar} \partial^{\nu} F_{\nu 0}(x)
    \end{align}
    and for the first moment, we have
    \begin{align}
        J_x(\mathbf{u_A}) \partial_0 \Psi(x)^* &= i \tilde{C}^{(M, t)}_\text{Scalar} \partial^{\nu} F_{\nu 0}(x) + i \tilde{C}^{(M, t)}_\text{Bilinear} \partial^{\nu} F_{\nu k}(x) \Gamma^k,\\
        J_x(\mathbf{u_A}) \partial_k \Psi(x)^* &= i \tilde{C}^{(M, r)}_\text{Scalar} \partial^{\nu} F_{\nu k}(x) + i \tilde{C}^{(M, r)}_\text{Bilinear} \partial^{\nu} F_{\nu 0}(x) \Gamma_k.
    \end{align}
    Thus, we conclude that we get a perturbation $\mathbf{u_{tot}}$ which preserves the restricted \ac{el} equation, if
    \begin{align}
        \diracad{\phi(x)} \gamma_\mu \phi(x) = \alpha \partial^{\nu} F_{\nu \mu}(x).
    \end{align}
    This shows that in the case of two fermion types, we recover the well-known Dirac Maxwell equation.

    \section{Conclusion and Outlook}\label{sec:conclusion-and-outlook}

    In this paper, we have shown that the CFS current given by Definition~\ref{def:cfs-current}
    is a good tool for analyzing the linearized field equations of causal fermion system.
    We showed how a perturbation can be constructed which preserves the restricted \ac{el} equations.
    In addition, the probing Theorem~\ref{thm:probing} in combination with the expansion Theorem~\ref{thm:expansion} in the
    continuum limit gives a direct relation between the linearized field equations of causal fermion system and the tensorial
    equations of motion of the Standard Model and general relativity.

    In Sections~\ref{sec:04-iepsilon-regularlized-minkowski-spacetime} and~\ref{sec:the-continuum-limit-of-minkowski-spacetime}
    we showed, that the construction gives rise to the equations of motion of the Standard Model.
    In this paper, we only outlined the derivation of axial current and the Maxwell-Dirac equation.
    Similar to~\cite{cfs} the method should also be applicable to the other equations of motion of the Standard Model, including
    non-abelian gauge fields.

    In particular, the procedure presented in this paper gives a clear pathway for future research.

    \subsection{Scalar and Pseudo-Scalar Fields}\label{subsec:scalar-and-pseudoscalar-fields}
    First of all, the analysis of a scalar perturbation (analogously to the vectorial perturbation presented in
    Section~\ref{subsec:maxwell-current}) is expect to derive field equations for real scalar fields.
    For instance, suitable perturbations could lead to equations describing the Higgs field, which would complete
    the derivation of the Standard Model from causal fermion systems in the continuum limit.
    Similarly, considering also pseudo-scalar contributions allows to derive equations of motion for
    complex scalar fields.

    \subsection{General Relativity}\label{subsec:general-relativity}
    The framework presented in this paper opens a new path towards deriving the Einstein field equations of General Relativity.
    This can be achieved by considering higher-order probing of the CFS current.
    Specifically, analyzing the vanishing of the second-order derivative of the CFS current with respect to the wave evaluation operator,
    \begin{align}
        \label{eq:einstein-equation}
        J_x(\mathbf{u}) \partial_{\mu} \partial_{\nu} \Psi(x) = 0,
    \end{align}
    is expected to encode the Einstein equations.
    The derivation of the Einstein equations can be done analogously to the procedure presented in this work.
    The claim that the resulting equations really give back the Einstein equations is strongly supported by~\cite[Section 4.9]{cfs}.
    This formulation would naturally incorporate gravitational effects directly from the fundamental dynamics of causal fermion systems,
    without requiring a pre-existing spacetime manifold.
    Evaluating~\eqref{eq:einstein-equation} requires a careful analysis of the continuum limit.
    It is expected that the leading contribution originates from the second iteration step in the derivation of the bosonic perturbations
    leading to quadratic factors of the vector potentials, whereas the contribution for the fermionic perturbation
    originates from the leading error terms of the approximation~\eqref{eq:localization-approximation}, i.e.
    \begin{align}
        \label{eq:expansion-of-fermionic-perturbation}
        D_{\mathbf{u}} P(x, y) \approx \frac{1}{2 \pi} \phi \diracad{\phi}
        + \frac{1}{2 \pi} \xi^{\mu} \left( \partial_\mu \phi \diracad{\phi} - \phi \diracad{\partial_\mu \phi } \right).
    \end{align}
    In addition, the metric contribution originates from a perturbation of the measure itself~\cite[Section 4.9]{cfs}.

    \subsection{Planck-Scale Corrections}\label{subsec:planck-scale-corrections}
    The systematic nature of our approach also allows for the derivation of Planck-scale corrections to the standard equations of motion.
    These corrections originate from the different error terms arising form the approximations done in the continuum limit
    (see Section~\ref{subsec:analysis-of-the-uv-divergence}):
    \begin{enumerate}
        \item Including the vectorial contribution $a_\mu \gamma^\mu$ to the closed chain (and therefore in the eigenvalues
        and the corresponding eigenspace projectors) gives rise to terms of orders $m \varepsilon$ in the derived equation for the perturbations.

        \item While our current analysis approximated fermionic fields and vector potentials as almost constant on the Planck scale,
        considering more refined approximations for the perturbations, such as~\eqref{eq:expansion-of-fermionic-perturbation},
        implies contributions of the energy-momentum tensor of order $m \varepsilon$.

        \item The error terms of the approximation done in~Appendix~\ref{subsec:computation-of-the-maxwell-current}
        to show the proportionality of the structure functions of Maxwell and Dirac perturbations also have to be studied more
        systematically, potentially leading to additional contributions of order $m \varepsilon$.
    \end{enumerate}
    All these contributions are expected to break the Lorentz invariance of the resulting equations.
    This lies in the nature of the chosen regularization scheme.
    However, a systematic analysis of different regularization schemes potentially allows to identify an optional regularization
    revealing information about the underlying quantumness of spacetime.

    \section{The Regularized Light-Cone Expansion}\label{sec:the-regularized-light-cone-expansion}.

    In this section, we proof Proposition~\ref{prop:electro-magnetic-solution}, which gives a formal solution for equation~\ref{eq:klein-gordon-equation-for-delta-t}.
    With the parameter $a \defeq m^2$, we rewrite this equation to
    \begin{align}
        \label{eq:klein-gordon-equation-for-delta-t-with-a}
        (\dAlembertian + a)
        \Delta T_a[A](x, y) = A(x) T^{(0)}_a(x, y).
    \end{align}
    We denote the derivatives of $T^{(0)}_a$ with respect to $a$ by
    \begin{align}
        T^{(n)}_a(x, y) &\defeq \pdv[n]{a} T^{(0)}_a (x, y) &\text{ for } n \in \N.
    \end{align}
    Note that $T^{(0)}_a(x, y)$ satisfies the Klein Gordan equation
    \begin{align}
    (\dAlembertian + a)
        T^{(0)}_a(x, y) = 0.
    \end{align}
    In addition, the derivatives $T^{(n)}_a$ have the following properties:

    \begin{lemma}
        \label{lem:tn-properties}
        Let $a > 0$ and $T_a(x, y)$ be defined as above, then $T^{(n)}_a(x, y)$ satisfies
        \begin{enumerate}
            \item $\displaystyle \partial^{(x)}_\mu T^{(n + 1)}_a(x, y) = \frac{\xi_\mu}{2} T^{(n)}_a(x, y)$,
            \item $\displaystyle (\dAlembertian^{(x)} + a) T^{(n + 1)}_a(x, y) = - (n + 1) T^{(n)}_a(x, y)$
        \end{enumerate}
        for every $n \in \N$.
    \end{lemma}
    \begin{proof}
        1) By definition, we have
        \begin{align*}
            T^{(0)}_a(x, y) &= \int \frac{\dd^4 k}{(2 \pi)^4} \delta(k^2 - a) \Theta(-k_0) e^{ik\xi} = \int \frac{\dd^4 k}{(2 \pi)^4} T^{(0)}_a(k) e^{ik\xi}\\
            T^{(n)}_a(x, y) &= \pdv[n]{a} T_a (x, y) = \int \frac{\dd^4 k}{(2 \pi)^4} (-1)^n \delta^{(n)}(k^2 - a) \Theta(-k_0) e^{ik\xi} \nonumber \\
            &= \int \frac{\dd^4 k}{(2 \pi)^4} T^{(n)}_a(k) e^{ik\xi}.
        \end{align*}
        Thus,
        \begin{align*}
            \pdv{k^\mu} T^{(n)}_a(k) = 2 k_\mu (-1)^n \delta^{(n + 1)}(k^2 - a) \Theta(-k_0) - (-1)^n \delta_\mu^0 \delta^{(n)}(k^2 - a) \delta(-k_0).
        \end{align*}
        However, since $a > 0$, the last term is always 0.
        Hence,
        \begin{align*}
            \pdv{k^\mu} T^{(n)}_a(x, y) = 2 k_\mu (-1)^n \delta^{(n + 1)}(k^2 - a) \Theta(-k_0) = - 2 k_\mu T^{(n + 1)}_a(k).
        \end{align*}
        and
        \begin{align*}
            \partial^{(x)}_\mu T^{(n + 1)}_a(x, y) &= \int \frac{\dd^4 k}{(2 \pi)^4} T^{(n + 1)}_a(k) (-i k_{\mu}) e^{ik\xi}\\
            &= \frac{i}{2} \int \frac{\dd^4 k}{(2 \pi)^4} \pdv{k^\mu} T^{(n)}_a(k) e^{ik\xi} \\
            &= - \frac{i}{2} \int \frac{\dd^4 k}{(2 \pi)^4} T^{(n)}_a(k) \pdv{k^\mu} e^{ik\xi}\\
            &= \frac{\xi_\mu}{2} \int \frac{\dd^4 k}{(2 \pi)^4} T^{(n)}_a(k) e^{ik\xi} = \frac{\xi_\mu}{2} T^{(n)}_a(x, y).
        \end{align*}

        \noindent 2) We know, that $(\dAlembertian^{(x)} + a) T^{(0)}_a(x, y) = 0$ and therefore $(k^2 - a) T^{(0)}_a(k) = 0$.
        Computing the $(n + 1)$-th derivative with respect to $a$ gives
        \begin{align*}
            0 &= \pdv[n + 1]{a} (k^2 - a) T^{(0)}_a(k) = (k^2 - a) T^{(n + 1)}_a(k) - (n + 1) T^{(n)}_a(k)\\
            &\Rightarrow (k^2 - a) T^{(n + 1)}_a(k) = (n + 1) T^{(n)}_a(k).
        \end{align*}
        The Fourier transformation of the above expression gives the proposed property.
    \end{proof}

    Property 1 of this lemma also justifies the notation $T^{(-1)}_{m^2}$ for the vector component of the kernel of the fermionic
    projector in~\eqref{eq:minkowski-vector-scalar-fermionic-projector}.
    We continue by checking that the formal series given in the Proposition~\ref{prop:electro-magnetic-solution}
    \begin{align}
        \Delta T_{a}[A](x, y) \defeq \sum_{n = 0}^{\infty} \frac{T_{a}^{(n + 1)}(x, y)}{n!} \int_0^1 (\tau - \tau^2)^n (\dAlembertian^n A)(x + \tau \xi) \dd \tau
    \end{align}
    is a solution of the perturbed Klein-Gordan~\eqref{eq:klein-gordon-equation-for-delta-t-with-a} in first order in $A$.

    \begin{proof}[Proof of Proposition~\ref{prop:electro-magnetic-solution}]
        First, we compute
        \begin{align*}
        (\dAlembertian^{(x)} &+ a)
            \Delta T_a[A](x, y) = \sum_{n = 0}^{\infty} \frac{1}{n!} \Bigg[ \\
            &(\dAlembertian^{(x)} + m^2) T^{(n + 1)}_a(x, y) \int_0^1 (\tau - \tau^2)^n (\dAlembertian^n A)(x + \tau \xi) \dd \tau \\
            &+ T^{(n + 1)}_a(x, y) \int_0^1 (\tau - \tau^2)^n \dAlembertian^{(x)} (\dAlembertian^n A)(x + \tau \xi) \dd \tau \\
            &+ 2 \partial^{(x)}_\mu T^{(n + 1)}_a(x, y) \int_0^1 (\tau - \tau^2)^n \partial^{(x) \mu} (\dAlembertian^n A)(x + \tau \xi) \dd \tau \Bigg].
        \end{align*}
        Applying the properties of elaborated in Lemma~\ref{lem:tn-properties} and the fact that
        \begin{align*}
            &\partial^{(x) \mu} (\dAlembertian^n A)(x + \tau \xi) = (1 - \tau) (\partial^\mu \dAlembertian^n A)(x + \tau \xi), \\
            &\Rightarrow \dAlembertian^{(x)} (\dAlembertian^n A)(x + \tau \xi) = (1 - \tau)^2 (\dAlembertian^{n + 1} A)(x + \tau \xi), \\
            &\xi_{\mu} (\partial^\mu \dAlembertian^n A)(x + \tau \xi) = \dv{\tau} (\dAlembertian^n A)(x + \tau \xi))
        \end{align*}
        then gives
        \begin{align*}
        (\dAlembertian^{(x)} &+ m^2)
            \Delta T[A](x, y) = \sum_{n = 0}^{\infty} \frac{1}{n!} \Bigg[ \\
            &- (n + 1) T^{(n)}(x, y) \int_0^1 (\tau - \tau^2)^n (\dAlembertian^n A)(x + \tau \xi) \dd \tau \\
            &+ T^{(n + 1)}(x, y) \int_0^1 (\tau - \tau^2)^n (1 - \tau)^2 (\dAlembertian^{n + 1} A)(x + \tau \xi) \dd \tau \\
            &+ T^{(n)}(x, y) \int_0^1 (\tau - \tau^2)^n (1 - \tau) \dv{\tau} (\dAlembertian^n A)(x + \tau \xi) \dd \tau \Bigg].
        \end{align*}
        Finally, integrating the last line by parts gives
        \begin{align*}
        (\dAlembertian^{(x)} &+ a)
            \Delta T[A]_a(x, y) = T^{(0)}_a(x, y) A(x) + \sum_{n = 0}^{\infty} \frac{1}{n!} \Bigg[ \\
            &- (n + 1) T^{(n)}_a(x, y) \int_0^1 (\tau - \tau^2)^n (\dAlembertian^n A)(x + \tau \xi) \dd \tau \\
            &+ T^{(n + 1)}_a(x, y) \int_0^1 (\tau - \tau^2)^n (1 - \tau)^2 (\dAlembertian^{n + 1} A)(x + \tau \xi) \dd \tau \\
            &- n T^{(n)}_a(x, y) \int_0^1 (\tau - \tau^2)^{n - 1} (1 - \tau)^2 (\dAlembertian^n A)(x + \tau \xi) \dd \tau \\
            &+ (n + 1) T^{(n)}_a(x, y) \int_0^1 (\tau - \tau^2)^{n - 1} (\dAlembertian^n A)(x + \tau \xi) \dd \tau \Bigg] \\
            &= T^{(0)}_a(x, y) A(x).
        \end{align*}
    \end{proof}

    Proposition~\ref{prop:electro-magnetic-solution} gives an exact solution of the perturbed~\eqref{eq:klein-gordon-equation-for-delta-t}.
    In the following, we show how to determine the perturbation of the kernel of the fermionic projector
    using this solution.
    We start with~\eqref{eq:minkowski-vector-scalar-fermionic-projector}
    \begin{align}
        P(x, y) = (i \slashed{\partial}^{(x)} + m) T^{(0)}_{m^2}(x, y).
    \end{align}
    then the defining equation for the perturbation $D_\mathbf{v} P$ is given by
    \begin{align}
    (i \slashed{\partial}^{(x)} - m)
        D_\mathbf{v} P(x, y) = A(x) P(x, y)
    \end{align}
    The inhomogeneous solution to this equation is then
    \begin{align}
        D_{\mathbf{v}} P(x, y) = \int \dd^4 z s(x, z) \slashed{A}(z) P(z, y)
    \end{align}
    where $s(x, z)$ is the Greens function of $i \slashed{\partial}^{(x)} - m$.
    Similar to $P(x, y)$, we can write
    \begin{align}
        s(x, z) = (i \slashed{\partial}^{(x)} + m) S(x, z)
    \end{align}
    where $S(x, z)$ is the Greens function of $\dAlembertian + m^2$.
    In addition, $T^{(0)}_{m^2}(x, y)$ satisfies
    \begin{align}
    (i \slashed{\partial}^{(x)} + m)
        T^{(0)}_{m^2}(x, y) = (-i \slashed{\partial}^{(y)} + m) T^{(0)}_{m^2}(x, y).
    \end{align}
    Thus,
    \begin{align}
        D_{\mathbf{v}} P(x, y) &= (i \slashed{\partial}^{(x)} + m) \gamma^\mu (-i \slashed{\partial}^{(y)} + m) \int \dd^4 z S(x, z) A_\mu(z) T^{(0)}_{m^2}(z, y)
    \end{align}
    By construction, $\int \dd^4 z S(x, z) A_\mu(z) T^{(0)}_{m^2}(z, y)$ is the inhomogeneous solution to equation~\ref{eq:klein-gordon-equation-for-delta-t}.
    By addition of a homogenous solution, we archive that $\left(D_{\mathbf{v}} P(x, y)\right)^* = D_{\mathbf{v}} P(y, x)$
    and thus
    \begin{align}
        D_{\mathbf{v}} P(x, y) = (i \slashed{\partial}^{(x)} + m) \gamma^\mu (-i \slashed{\partial}^{(y)} + m) \Delta T_{m^2}[A_\mu](x, y)
    \end{align}
    where $\Delta T_{m^2}[A_\mu](x, y)$ is given by the formal series in Proposition~\ref{prop:electro-magnetic-solution}.
    The first terms of $D_{\mathbf{v}} P(x, y)$ are then given by
    \begin{align}
        D_{\mathbf{v}} P(x, y) &= \frac{1}{2} \slashed{\xi} \int_0^1 A_{\mu}(x + \tau \xi) \xi^\mu \dd \tau T^{(-1)}_{m^2}(x, y) \nonumber \\
        &~~~- \frac{1}{2} \slashed{\xi} \xi^{\mu} \int_0^1 (\tau - \tau^2) (\partial^\nu F_{\mu \nu})(x + \tau \xi) \dd \tau T^{(0)}_{m^2}(x, y) \nonumber \\
        &~~~+ \frac{1}{4} \slashed{\xi} \gamma^{\mu} \gamma^{\nu} \int_0^1 F_{\mu \nu}(x + \tau \xi) \dd \tau T^{(0)}_{m^2}(x, y) \nonumber \\
        &~~~- \xi^\mu \gamma^\nu \int_0^1 (1 - \tau) F_{\mu \nu}(x + \tau \xi) \dd \tau T^{(0)}_{m^2}(x, y) \nonumber \\
        &~~~- \gamma^{\mu} \int_0^1 (1 - \tau)^2 (\partial^\nu F_{\mu \nu})(x + \tau \xi) \dd \tau T^{(1)}_{m^2}(x, y),
    \end{align}
    where $F_{\mu\nu}(x) \defeq \partial_\mu A_\nu(x) - \partial_\nu A_\mu(x)$ is the usual electro magnetic field tensor

    \section{Computation of the CFS Current in the Continuum Limit}\label{sec:computation-of-the-cfs-current}

    In this section, we outline the calculations of the current contribution for the Dirac and Maxwell perturbation.
    All the calculations are done in the continuum limit, i.e.~we work in the setting derived in Section~\ref{subsec:analysis-of-the-uv-divergence}
    In this case, we have
    \begin{align*}
        P(x, y) &= \frac{i}{2} T^{(-1)}_m(x, y) \slashed{\xi}, \\
        \lambda_{\pm} &= \frac{\abs{T^{(-1)}_m(x, y)}^2}{4} \left( \xi^\mu \complexconj{\xi_\mu} \pm 2 i \varepsilon r \right), \\
        \Lambda_{(\pm, L/R)} &= \frac{1}{4} \left( 1 \mp_{L/R} \gamma^5 \right) \left( 1 \pm i\frac{\xi_k}{r} \Gamma^k \right).
    \end{align*}

    \subsection{Dirac Current}\label{subsec:computation-of-the-dirac-current}

    For the Dirac current we consider the perturbation
    \begin{align*}
        D_{\mathbf{u}} P(x, y) = \frac{1}{2 \pi} \phi(x) \diracad{\phi(y)} \approx \frac{1}{2\pi}  \phi\left(\frac{x + y}{2}\right) \diracad{\left(\frac{x + y}{2}\right)}
    \end{align*}
    A direct computation then gives
    \begin{align}
        \label{eq:dirac-current-perturbation}
        \Re \tr \left[ \complexconj{\lambda_\pm} \Lambda_{(\pm, L/R)} D_\mathbf{u} A_{xy} \right] &= f^{(D)}_s(t, r) j_v^0 + f^{(D)}_a(t, r) j_v^r \nonumber \\
        &~~~\pm \mp_{L/R} \left(f^{(D)}_s(t, r) j_a^r + f^{(D)}_a(t, r) j_a^0\right)
    \end{align}
    with the real valued functions
    \begin{align}
        f^{(D)}_s(t, r) &= \frac{1}{8 \pi} \abs{T^{(-1)}_m(x, y)}^2 \left( 2 \varepsilon r^2 \Re \left[ T^{(-1)}_m(t, r) \right] - \xi^\mu \complexconj{\xi_\mu} \Im \left[ T^{(-1)}_m(t, r) \xi^0 \right] \right) \\
        f^{(D)}_a(t, r) &= \frac{1}{8 \pi} \abs{T^{(-1)}_m(x, y)}^2 \left( 2 \varepsilon r \Re \left[ T^{(-1)}_m(t, r) \xi^0 \right] - r \xi^\mu \complexconj{\xi_\mu} \Im \left[ T^{(-1)}_m(t, r) \right] \right)
    \end{align}
    and the usual Dirac currents
    \begin{align}
        j_v^{\mu}(x) \approx j_v^{\mu}\left(\frac{x + y}{2}\right) &= \diracad{\phi\left(\frac{x + y}{2}\right)} \gamma^{\mu} \phi\left(\frac{x + y}{2}\right), \\
        j_a^{\mu}(x) \approx j_a^{\mu}\left(\frac{x + y}{2}\right) &= \diracad{\phi\left(\frac{x + y}{2}\right)} \gamma^5 \gamma^{\mu} \phi\left(\frac{x + y}{2}\right).
    \end{align}
    The functions $f^{(D)}_s(t, r)$ and $f^{(D)}_a(t, r)$ describe the non-local part of the perturbation whereas the vector fields
    $j_v$ and $j_a$ the local contribution.
    The subscript $s$ and $a$ of the functions $f^{(D)}_s(t, r)$ and $f^{(D)}_a(t, r)$ stand for symmetric and antisymmetric
    with respect to a time flip, i.e.~the functions have the properties
    \begin{align}
        f^{(D)}_s(-t, r) = f^{(D)}_s(t, r) && \text{ and } && f^{(D)}_a(-t, r) = -f^{(D)}_a(t, r).
    \end{align}
    These characterizations are very useful to identify canceling terms in the integration of the current.
    However, before we continue with the integration, we first analyze the perturbation corresponding to a vector potential.

    \subsection{Maxwell Current}\label{subsec:computation-of-the-maxwell-current}

    For the Maxwell current we distinguish between an axial perturbation
    \begin{align}
        D_{\mathbf{u}} P(x, y) = (i \slashed{\partial}^{(x)} + m) \gamma^5 \gamma^\mu (-i \slashed{\partial}^{(y)} + m) \Delta T_{m^2}[A_\mu](x, y)
    \end{align}
    and a vectorial perturbation
    \begin{align}
        D_{\mathbf{u}} P(x, y) = (i \slashed{\partial}^{(x)} + m) \gamma^\mu (-i \slashed{\partial}^{(y)} + m) \Delta T_{m^2}[A_\mu](x, y).
    \end{align}
    A direct computation then gives
    \begin{align}
    \end{align}
    \begin{align}
        \Re \tr \left[ \complexconj{\lambda_\pm} \Lambda_{(\pm, L/R)} D_\mathbf{u} A_{xy} \right] &= f^{(M)}_s(t, r) \partial_\mu F^{\mu 0} + f^{(M)}_a(t, r) \frac{\xi_i}{r} \partial_\mu F^{\mu i} \nonumber \\
        &\text{ for vectorial perturbation}, \\
        \Re \tr \left[ \complexconj{\lambda_\pm} \Lambda_{(\pm, L/R)} D_\mathbf{u} A_{xy} \right] &= \pm \mp_{L/R} \left( f^{(M)}_s(t, r) \partial_\mu F^{\mu 0} + f^{(M)}_a(t, r) \frac{\xi_i}{r} \partial_\mu F^{\mu i} \right) \nonumber \\
        &\text{ for axial perturbation}
    \end{align}
    with the real valued functions
    \begin{align}
        f^{(M)}_s(t, r) &= - \frac{\abs{T^{(-1)}_m(x, y)}^2}{12} \bigg( 2 \varepsilon r^2 \Re \left[ \complexconj{T^{(1)}_m(x, y)} T^{(-1)}_m(x, y) \right] \nonumber \\
        &~~~ - \complexconj{\xi_\mu} \xi^{\mu} \Im \left[ \complexconj{T^{(1)}_m(x, y)} T^{(-1)}_m(x, y) \xi^0 \right] \bigg)
        - \frac{\abs{\lambda}^2}{12} \Im \left[ \frac{T^{(0)}_m(x, y)}{T^{(-1)}_m(x, y)} \xi^0 \right], \\
        f^{(M)}_a(t, r) &= -\frac{\abs{T^{(-1)}_m(x, y)}^2}{12} \bigg( 2 \varepsilon r \Re \left[ \complexconj{T^{(1)}_m(x, y)} T^{(-1)}_m(x, y) \xi^0 \right] \nonumber \\
        &~~~ -r \complexconj{\xi_\mu} \xi^{\mu} \Im \left[ \complexconj{T^{(1)}_m(x, y)} T^{(-1)}_m(x, y) \right] \bigg)
        - \frac{\abs{\lambda}^2}{12} r \Im \left[ \frac{T^{(0)}_m(x, y)}{T^{(-1)}_m(x, y)} \right].
    \end{align}
    Having determined the structure function for the Maxwell perturbation, we can now apply the same simplifaction rules as
    for the Dirac current to calculate the CFS current coefficients.

    The real coefficients of the corresponding Dirac matrics in the CFS current are up to an overall real
    factor equal to the one of the Dirac current.
    This can be seen as follows.
    First observe that
    \begin{align}
        \abs{\lambda} = \frac{\abs{T^{(-1)}_m(x, y)}^2}{4 m^2} \abs{z^2}
    \end{align}
    where $z = m \sqrt {- \xi_\mu \xi^\mu}$ is the argument of the Bessel functions in $T^{(n)}$ (see~\eqref{eq:def-bessel-function-argument}).
    Using this relation, we get
    \begin{align}
        \abs{\lambda}^2 \Im \left[ \frac{T^{(0)}_m(x, y)}{T^{(-1)}_m(x, y)} \xi^0 \right] &= -\frac{\abs{T^{(-1)}_m(x, y)}^2}{16 m^2} \bigg( 2 \varepsilon r^2 \Re \left[ \complexconj{z^2 T^{(0)}_m(x, y)} T^{(-1)}_m(x, y) \right]\nonumber \\
        &~~~- \xi^{\mu} \complexconj{\xi_\mu} \Im \left[ \complexconj{z^2 T^{(0)}_m(x, y)} T^{(-1)}_m(x, y) \xi^0 \right] \bigg), \\
        \abs{\lambda}^2 r \Im \left[ \frac{T^{(0)}_m(x, y)}{T^{(-1)}_m(x, y)} \right] &= -\frac{\abs{T^{(-1)}_m(x, y)}^2}{16 m^2} \bigg( 2 \varepsilon r \Re \left[ \complexconj{z^2 T^{(0)}_m(x, y)} T^{(-1)}_m(x, y) \xi^0 \right] \nonumber \\
        &~~~- \xi^{\mu} \complexconj{\xi_\mu} r \Im\left[ \complexconj{z^2 T^{(0)}_m(x, y)} T^{(-1)}_m(x, y) \right] \bigg).
    \end{align}
    The function $z^2 T^{(0)}_m(x, y)$ is analytic and in particular, smooth $x = y$.
    Thus, similar to the argument given above for the particle wavefunction and the vector potential, we approximate this function at $x = y$.
    This leads to a real factor
    \begin{align}
        z^2 T^{(0)}_m(x, y) \approx \left. z^2 T^{(0)}_m(x, y) \right|_{y = x} = \frac{m^2}{8 \pi^3}.
    \end{align}

    The function $T^{(1)}_m(x, y)$ contains a logarithmic pole on the light-cone.
    As shown in~\cite[Section~3.7]{cfs}, this pole can be absorbed by a so-called \textit{microlocal chiral transformation}, which
    is a again a certain perturbation of the kernel of the fermionic projector.
    The consideration of this perturbation allows us to work again with a smooth function instead of the $T^{(1)}_m(x, y)$.
    In addition, this perturbation requires the existence of at least three generations.
    The interested reader is referred to~\cite{cfs} for more details.
    For this paper, we just replace $T^{(1)}_m(x, y)$ by the smooth function $g(x, y) \approx g(x, x) \eqdef c \in \R$.
    Thus, we get
    \begin{align}
        \label{eq:structure-function-relations}
        f^{(M)}_s(t, r) = \left( \frac{1}{196 \pi^2} - \frac{2\pi}{3} c \right) f^{(D)}_s(t, r),\\
        f^{(M)}_a(t, r) = \left( \frac{1}{196 \pi^2} - \frac{2\pi}{3} c \right) f^{(D)}_a(t, r).
    \end{align}
    Thus, the structure functions of the Maxwell perturbation are up to an overall real factor equal to the ones of the Dirac perturbation.

    \subsection{Analysis of the Current Coefficients}\label{subsec:analysis-of-the-current-coefficients}

    In the CFS-current, the only term depending on the perturbation is the term
    \begin{align}
        \Re \tr \left[ \complexconj{\lambda_i} \Lambda_{i} D_\mathbf{u} A_{xy} \right]
    \end{align}
    which as shown in the previous two sections has the form as given in~\eqref{eq:dirac-current-perturbation}.
    The structure functions $f_s(t, r)$ and $f_a(t, r)$ define how the standard currents $j^\mu_{v/a}$ are related to
    the perturbation.

    Since we integrate over a symmetric domain with respect to the time coordinate, antisymmetric contributions cancel.
    Similarly, these functions are spherically symmetric as they only depend on the radial coordinate.
    Thus, the $S^2$ integration simplifies terms containing $\xi^i$ terms
    \begin{align}
        \int_{S^2} \xi^i \dd \Omega = 0 && \text{ and } && \int_{S^2} \frac{\xi^i \xi^j}{r^2} \dd \Omega = \frac{4 \pi}{3} \delta^{ij}
    \end{align}

    Applying these simplification rules to the expression occurring in $J_x (\mathbf{u}) \Psi(x)^*$ we get
    for one sector
    \begin{align}
        J_x (\mathbf{u}) \Psi(x)^* &= \frac{i}{3} C^{(0)} j^k_a \gamma^5 \Gamma_k
    \end{align}
    with the real constant
    \begin{align}
        C^{(0)} &= 4 \pi \int_{-\infty}^\infty \dd t \int_0^{\infty} \dd r ~ r^2 f_s(t, r)
    \end{align}
    For the CFS-Current with one derivative, we use Lemma~\ref{lem:tn-properties} to get
    \begin{align}
        \Psi(y) \partial_{\mu} \Psi(x) = \partial^{(x)}_{\mu} P(y, x) = \frac{\complexconj{\xi_\mu}}{2} \frac{\complexconj{T^{(-2)}_m(x, y)}}{\complexconj{T^{(-1)}_m(x, y)}} P(y, x)
        + \frac{i}{2} \complexconj{T^{(-1)}_m(x, y)} \gamma_{\mu}
    \end{align}
    Therefore, we have
    \begin{align}
        J_x (\mathbf{u}) \partial^0 \Psi(x)^* &= i C^{(t)}_{\text{Pseudo-Scalar}} ~ j_a^0 ~ \gamma^5 \nonumber \\
        &~~~ + C^{(t)}_{\text{Pseudo-Bilinear}} ~ j_a^k ~ \gamma^5 \Gamma_k,\\
        J_x (\mathbf{u}) \partial^k \Psi(x)^* &= i C^{(r)}_{\text{Pseudo-Scalar}} ~ j_a^k ~ \gamma^5 \nonumber \\
        &~~~ + C^{(r)}_{\text{Pseudo-Bilinear}} ~ j_a^0 \gamma^5 \Gamma^k
    \end{align}
    with
    \begin{align}
        C^{(t)}_{\text{Pseudo-Scalar}} &= -\frac{\pi}{2} \int_{-\infty}^\infty \dd t \int_0^{\infty} \dd r \frac{\abs{T^{(-1)}_m(x, y)}^4}{\abs{\lambda}^2} f_a(t, r) t r^3 \varepsilon \\
        C^{(r)}_{\text{Pseudo-Scalar}} &= -\frac{\pi}{12} \int_{-\infty}^\infty \dd t \int_0^{\infty} \dd r \frac{\abs{T^{(-1)}_m(x, y)}^4}{\abs{\lambda}^2} f_s(t, r) \left( t^2 + r^2 + \varepsilon^2 \right) r^2 \varepsilon \\
        C^{(t)}_{\text{Pseudo-Bilinear}} &= -\frac{\pi}{12} \int_{-\infty}^\infty \dd t \int_0^{\infty} \dd r \frac{\abs{T^{(-1)}_m(x, y)}^4}{\abs{\lambda}^2} f_s(t, r) \left( t^2 + r^2 + \varepsilon^2 \right) r^2 \varepsilon \nonumber \\
        &~~~ + \frac{2 \pi}{3} \int_{-\infty}^\infty \dd t \int_0^{\infty} \dd r f_s(t, r) r^2 \Im\left[ \frac{T^{(-2)}_m(x, y)}{T^{(-1)}_m(x, y)} \xi^0 \right] \nonumber \\
        &= C^{(r)}_{\text{Pseudo-Scalar}} + \frac{2 \pi}{3} \int_{-\infty}^\infty \dd t \int_0^{\infty} \dd r f_s(t, r) r^2 \Im \left[ \frac{T^{(-2)}_m(x, y)}{T^{(-1)}_m(x, y)} \xi^0 \right]
        \\
        C^{(r)}_{\text{Pseudo-Bilinear}} &= 2 \pi \int_{-\infty}^\infty \dd t \int_0^{\infty} \dd r f_a(t, r) \bigg( \frac{r^3}{3} \Im \left[ \frac{T^{(-2)}_m(x, y)}{T^{(-1)}_m(x, y)} \right] \nonumber \\
        &~~~- \frac{\abs{T^{(-1)}_m(x, y)}^4}{4 \abs{\lambda}^2} t r^3 \varepsilon \bigg) \nonumber \\
        &= C^{(t)}_{\text{Pseudo-Scalar}} + \frac{2 \pi}{3} \int_{-\infty}^\infty \dd t \int_0^{\infty} \dd r f_a(t, r) r^3 \Im \left[ \frac{T^{(-2)}_m(x, y)}{T^{(-1)}_m(x, y)} \right]
    \end{align}

    For two sectors, where the perturbation only acts on one of the two sectors, we get
    \begin{align}
        J_x (\mathbf{u}) \Psi(x)^* &= 2 C^{(0)} j^k_v + \frac{5i}{6} C^{(0)} j^k_a \gamma^5 \Gamma_k
    \end{align}
    and
    \begin{align}
        J_x (\mathbf{u}) \partial_0 \Psi(x)^* &= i C^{(t)}_{\text{Scalar}} j_v^0 + C^{(t)}_{\text{Bilinear}} j_v^k \Gamma_k + \frac{5}{2} \text{(axial current terms)}, \\
        J_x (\mathbf{u}) \partial^k \Psi(x)^* &= i C^{(r)}_{\text{Scalar}} j_v^k + C^{(r)}_{\text{Bilinear}} j_v^0 \Gamma_k + \frac{5}{2} \text{(axial current terms)}
    \end{align}
    with
    \begin{align}
        C^{(t)}_{\text{Scalar}} &= - \frac{\pi}{2} \int_{-\infty}^\infty \dd t \int_0^{\infty} \dd r f_s(t, r) \frac{\abs{T^{(-1)}_m(x, y)}^4}{\abs{\lambda}^2} (t^2 + r^2 + \varepsilon^2 ) r^2 \varepsilon \nonumber \\
        &~~~- 4 \pi \int_{-\infty}^\infty \dd t \int_0^{\infty} \dd r f_s(t, r) r^2 \Im\left[ \frac{T^{(-2)}_m(x, y)}{T^{(-1)}_m(x, y)} \xi^0 \right] \\
        C^{(r)}_{\text{Scalar}} &= -\frac{\pi}{3} \int_{-\infty}^\infty \dd t \int_0^{\infty} \dd r f_a(t, r) \frac{\abs{T^{(-1)}_m(x, y)}^4}{\abs{\lambda}^2} t r^3 \varepsilon \nonumber \\
        &~~~- \frac{4 \pi}{3} \int_{-\infty}^\infty \dd t \int_0^{\infty} \dd r f_a(t, r) r^3 \Im \left[ \frac{T^{(-2)}_m(x, y)}{T^{(-1)}_m(x, y)} \right]\\
        C^{(t)}_{\text{Bilinear}} &= \frac{\pi}{24} \int_{-\infty}^\infty \dd t \int_0^{\infty} \dd r \frac{\abs{T^{(-1)}_m(x, y)}^4}{\abs{\lambda}^2} f_s(t, r) \left( t^2 + r^2 + \varepsilon^2 \right) r^2 \varepsilon \nonumber \\
        &~~~+ \frac{4 \pi}{3} \int_{-\infty}^\infty \dd t \int_0^{\infty} \dd r \Im\left[ \frac{T^{(-2)}_m(x, y)}{T^{(-1)}_m(x, y)} \xi^0 \right] f_s(t, r) r^2 \\
        C^{(r)}_{\text{Bilinear}} &= - \frac{\pi}{2} \int_{-\infty}^\infty \dd t \int_0^{\infty} \dd r f_s(t, r) \frac{\abs{T^{(-1)}_m(x, y)}^4}{\abs{\lambda}^2}  \varepsilon \left( t^2 + r^2 + \varepsilon^2 \right) r^2
    \end{align}

    Note that all these constants depend linearly on the structure functions $f_s(t, r)$ and $f_a(t, r)$.
    When determining the coupling constants between the contributions of the different perturbations only the ratios
    are relevant.
    Due to~\eqref{eq:structure-function-relations} we get that all ratios of these constants are equal to
    \begin{align}
        \frac{1}{\alpha} \defeq \frac{1}{196 \pi^2} - \frac{2 \pi}{3} c.
    \end{align}

    \printbibliography
\end{document}